\documentclass[12pt,a4paper,notitlepage]{article}
\usepackage{amsmath}
\usepackage{amssymb}
\usepackage{amsfonts}
\usepackage[onehalfspacing,nodisplayskipstretch]{setspace}
\usepackage[a4paper,nohead,left=1in,right=1in,top=1in,bottom=1in]{geometry}
\usepackage{hyperref}

\newtheorem{theorem}{Theorem}

\newtheorem{corollary}[theorem]{Corollary}

\newtheorem{example}[theorem]{Example}
\let\oldexample\example
\renewcommand{\example}{\oldexample\normalfont}

\newtheorem{lemma}[theorem]{Lemma}

\newtheorem{proposition}[theorem]{Proposition}

\newenvironment{proof}[1][Proof]{\noindent\textbf{#1.} }{\hfill $\square$}

\input{tcilatex}
\begin{document}

\title{Proper Calibeating\thanks{%
Previous versions: December 2025; May 2026 (\texttt{arXiv:2605.26703v1});
June 2026 (\texttt{arXiv:2605.26703v2}). Latest version may be downloaded
from \url{http://ma.huji.ac.il/hart/abs/calib-proper.html}. A general
presentation that includes results from this paper is available at %
\url{http://ma.huji.ac.il/hart/pres.html\#calib-beat-p}.}}
\author{Dean P. Foster\thanks{%
Department of Statistics, Wharton, University of Pennsylvania, Philadelphia,
and Amazon, New York. \emph{e-mail}: \texttt{dean@foster.net}. \ \emph{web
page}: \texttt{\url{http://deanfoster.net}.}} \and Sergiu Hart\thanks{%
Institute of Mathematics, Department of Economics, and Federmann Center for
the Study of Rationality, The Hebrew University of Jerusalem. \emph{e-mail}: 
\texttt{hart@huji.ac.il}. \ \emph{web page}: \texttt{%
\url{http://ma.huji.ac.il/hart}.}}}
\maketitle

\begin{abstract}
The classic concept of \textquotedblleft calibrated
forecasts\textquotedblright\ and its more recent refinement,
\textquotedblleft calibeating,\textquotedblright\ are defined with respect
to the standard quadratic scoring rule. We extend these notions to the class
of proper scoring rules (for which the true distribution is an optimal
forecast) and define \emph{proper calibration} and \emph{proper calibeating}
by requiring the corresponding guarantees to hold uniformly over all bounded
proper scoring rules. We first establish that calibration always implies
proper calibration, whereas calibeating need not imply proper calibeating.
Second, we show how to guarantee proper calibeating and proper
multicalibeating; in particular, \emph{complete calibeating}---a strong form
of calibeating that calibeats the joint binning---is always proper. Finally,
we consider \emph{decision-making under uncertainty}, where one best replies
to the forecasts. We establish that (proper) calibration is equivalent to
universal no regret, and (proper) complete calibeating is equivalent to
universal no regret together with universal subsuming of the reference
sequence, where \emph{universal}\ refers to all bounded utility functions.
\end{abstract}

\tableofcontents

%TCIMACRO{%
%\TeXButton{References Without Numbers}{\def\@biblabel#1{#1\hfill}
%\def\thebibliography#1{\section*{References}
%\addcontentsline{toc}{section}{References}
%\list
%{}{
%\labelwidth 0pt
%\leftmargin 1.8em
%\itemindent -1.8em
%\usecounter{enumi}}
%\def\newblock{\hskip .11em plus .33em minus .07em}
%\sloppy\clubpenalty4000\widowpenalty4000
%\sfcode`\.=1000\relax\def\baselinestretch{1}\large \normalsize}
%\let\endthebibliography=\endlist} }%
%BeginExpansion
\def\@biblabel#1{#1\hfill}
\def\thebibliography#1{\section*{References}
\addcontentsline{toc}{section}{References}
\list
{}{
\labelwidth 0pt
\leftmargin 1.8em
\itemindent -1.8em
\usecounter{enumi}}
\def\newblock{\hskip .11em plus .33em minus .07em}
\sloppy\clubpenalty4000\widowpenalty4000
\sfcode`\.=1000\relax\def\baselinestretch{1}\large \normalsize}
\let\endthebibliography=\endlist
%EndExpansion
%TCIMACRO{%
%\TeXButton{define shfrac}{\newcommand{\shfrac}[2]{\ensuremath{{}^{#1} \hspace{-0.04in}/_{\hspace{-0.03in}#2}}}} }%
%BeginExpansion
\newcommand{\shfrac}[2]{\ensuremath{{}^{#1} \hspace{-0.04in}/_{\hspace{-0.03in}#2}}}
%EndExpansion
%TCIMACRO{%
%\TeXButton{def \T \B}{\newcommand\T{\rule{0pt}{2.6ex}}
%\newcommand\B{\rule[-1.2ex]{0pt}{0pt}}}}%
%BeginExpansion
\newcommand\T{\rule{0pt}{2.6ex}}
\newcommand\B{\rule[-1.2ex]{0pt}{0pt}}%
%EndExpansion

\section{Introduction\label{s:intro}}

Forecasting the probability of future events is a foundational problem.
Forecasters issue probabilistic predictions that are then compared with the
outcomes that materialize. How should one evaluate such a forecaster? A
natural approach is to use a \emph{scoring rule}: a penalty function that
assigns a loss to each forecast--outcome pair. A scoring rule is deemed 
\emph{proper} if the loss is minimized when the forecast equals the
underlying probability distribution of the outcome. The standard and most
widely used proper scoring rule is the \emph{quadratic} rule (Brier 1950;
see Savage 1971, Schervish 1989, and the comprehensive treatment in Gneiting
and Raftery 2007).

A classic result (Sanders 1963; Murphy 1972; DeGroot and Fienberg 1983) is
that the Brier quadratic score $\mathcal{B}$ decomposes as 
\begin{equation*}
\mathcal{B}=\mathcal{K}+\mathcal{R},
\end{equation*}%
where $\mathcal{K}$ is the \emph{calibration} score---measuring how close
the forecasts are to the realized frequencies---and $\mathcal{R}$ is the 
\emph{refinement} score---measuring how informative the partitioning of the
outcomes into bins is, where the partitioning is determined by the announced
forecasts (see Section \ref{sus:sequences} for the precise setup and
definitions). A surprising result of Foster and Vohra (1998) is that one can
generate forecasts that are \emph{guaranteed} to be calibrated (i.e., $%
\mathcal{K}$ converges to zero), no matter what the outcomes turn out to be
(for the ensuing literature, see the survey of Olszewski 2015).\footnote{%
For work on calibration that appeared after this survey, see Foster and Hart
(2018, 2021, 2023), Hart (2025), Marx, Kuleshov, and Ermon (2024),
Okoroafor, Sun, and Kleinberg (2024), Qiao and Valiant (2021).}

In Foster and Hart (2023), we introduced the notion of \textquotedblleft
calibeating\textquotedblright : beating a reference forecaster\footnote{%
The \textquotedblleft reference\textquotedblright\ forecaster could also be
the forecaster's own alternative procedure (based, for instance, on certain
knowledge or information that he possesses); in this case, calibeating
entails incorporating this knowledge into the new forecast.} by achieving a
Brier score no worse than its refinement score---thus gaining calibration
without losing expertise, which, as we argued there, is reflected in the
induced partitioning into bins and the resulting refinement score. We showed
that calibeating can be achieved in several ways: by a simple deterministic
online procedure; by a stochastic procedure that is itself calibrated; and
by a deterministic continuously calibrated procedure. All these results,
however, were stated exclusively for the quadratic score.\footnote{%
In Appendix A.9 of the full version of the paper, Foster and Hart (2023f),
we show how a modification of our simple procedure yields calibeating with
respect to the logarithmic proper scoring rule.}

However, there is an extensive class of proper scoring rules (in particular,
every decision-making setting under uncertainty generates one; see below and
Section \ref{s:utility}). A fundamental concern is thus whether these
guarantees are robust or merely artifacts of the quadratic functional form.
If a forecaster's performance evaluation is sensitive to the specific choice
of a proper loss function, its theoretical and practical relevance may well
be questionable. This raises the natural question:

\begin{quote}
\emph{If a procedure is calibrated (or calibeats) under the quadratic score,
do these guarantees persist under every other proper scoring rule?}
\end{quote}

This paper answers that question for proper scoring rules that are bounded.%
\footnote{%
This excludes the logarithmic scoring rule. Some of the proper calibeating
results are limited to Lipschitz proper scoring rules.} The answer turns out
to be positive for calibration and negative for calibeating.

We will call a procedure \emph{proper calibrated} (respectively, \emph{%
proper calibeating}) if the corresponding guarantee holds \emph{%
simultaneously} for all bounded proper scoring rules (see Section \ref%
{s:proper} for the precise definitions).

\textbf{Calibration entails proper calibration. }We show, first, that every
calibrated procedure is automatically proper calibrated. The key observation
is that the calibration score under any bounded proper scoring rule is
bounded by a multiple of the square root of the quadratic calibration score.
Consequently, all established calibrated forecasting procedures---the
stochastic procedures, as well as the deterministic ones that are
continuously calibrated\footnote{%
See Foster and Hart (2021), with precursors Kakade and Foster (2004), Foster
and Kakade (2006), Foster and Hart (2018).}---are already proper, with no
modification needed.

\textbf{Calibeating does \emph{not}} \textbf{entail proper calibeating. }Our
second result is perhaps a surprise: unlike calibration, calibeating does 
\emph{not} transfer to proper calibeating. We exhibit a concrete example
(Section \ref{sus:example}) in which a forecasting sequence $\mathbf{c}$
calibeats a reference sequence $\mathbf{b}$ under the quadratic scoring
rule, yet \emph{fails} to calibeat $\mathbf{b}$ under, for instance, the $2$%
-spherical proper scoring rule (which is bounded and even Lipschitz). Both
sequences $\mathbf{b}$ and $\mathbf{c}$ are perfectly\footnote{%
Throughout, the term \textquotedblleft perfect\textquotedblright\ refers to
the ideal situation where there are no errors; thus, perfect calibration
means that forecasts and state averages always coincide.} calibrated; the
failure is purely in the refinement comparison.

The reason for the asymmetry between calibration and calibeating is
instructive. Calibration scores shrink whenever the quadratic calibration
score shrinks, regardless of which proper scoring rule is used; but
refinement scores under different proper rules need not move in tandem.

\textbf{Achieving proper calibeating. }We next provide three positive
results on proper calibeating, i.e., calibeating simultaneously under the
appropriate class of proper scoring rules (bounded or Lipschitz).

\emph{First,} we show that calibeating the \emph{joint binning} that
cross-classifies forecasts by the bins of both the reference forecaster and
the procedure itself---we call this \emph{complete calibeating}\footnote{%
This turns out to be a useful strengthening of calibeating; see Sections \ref%
{sus:joint} and \ref{susus:u-subsuming}.}---is proper (i.e., robust across
all bounded proper scoring rules), and thus yields both proper calibeating
and proper calibration. A key insight is that it is not enough to calibeat
the reference forecaster and to be calibrated separately; what matters is
calibeating the joint binning, which keeps track of how the two sets of bins
interact (Appendix \ref{sus-a:bxc} elaborates further on this).

\emph{Second,} the very simple deterministic procedure of Theorem 3 of
Foster and Hart (2023)---which forecasts the running state average in the
reference forecaster's bin---is shown to be calibeating for every Lipschitz
proper scoring rule (which we term \textquotedblleft \emph{proper-Li}
calibeating\textquotedblright ). We further show that this result does not
hold for all bounded proper scoring rules.

\emph{Third,} we establish the existence of a deterministic procedure that
is simultaneously proper-Li calibeating and proper continuously-calibrated;
this is obtained via an approximate decomposition of the score (a
generalization of the quadratic case; see Appendix A.7 in Foster and Hart
2023ff).

All three results extend to multiple reference forecasters by considering
the joint binning that cross-classifies all of their bins.

\textbf{Decision-making under uncertainty.} Forecasts are commonly used in
decision-making under uncertainty, where the probabilities of the various
states of nature are unknown (see Section \ref{s:utility}). The forecast is
used to choose a decision that maximizes the decision maker's expected
utility. This process induces a scoring rule in which the loss is the
realized disutility from the decision that is optimal given the forecast.
The scoring rule is proper: when the forecast is correct,\ i.e., equal to
the true distribution, expected utility is maximized and expected loss is
minimized. We show that the calibration score of the sequence of forecasts,
computed according to the induced scoring rule, is exactly the
\textquotedblleft regret\textquotedblright\ from best replying to these
forecasts. Consequently, proper calibration is equivalent to achieving no
regret \textquotedblleft universally,\textquotedblright\ i.e.,
simultaneously across all bounded utility functions. When given a reference
sequence, proper complete calibeating is shown to be equivalent to universal
no regret together with universal subsuming of the decision value of the
reference sequence (which can thus be ignored with no loss in utility).

\subsection{Related Work\label{sus:intro-related}}

The question of how forecasting guarantees extend across different scoring
rules has attracted attention from several directions.

Universality across losses is central to U-calibration (Kleinberg, Leme,
Schneider, and Teng 2023) and omniprediction (Gopalan, Kalai, Reingold,
Sharan, and Wieder 2022). Hu and Wu (2024) give a decision-theoretic
calibration measure based on regret for all bounded payoffs. Lee, Noarov,
Pai, and Roth (2022) study multicalibeating through online multiobjective
optimization. Chen, Huang, Jordan, and Luo (2026) and, independently,
Fichtl, Guzm\'{a}n, and Mehta (2026) develop calibeating for general proper
losses through regret-minimization methods; the latter use a
Bregman-divergence approach and obtain simultaneous guarantees for several
classes of proper losses. On the statistical side, Dimitriadis, Gneiting,
and Jordan (2021) and Popordanoska, Gruber, Tiulpin, Buettner, and Blaschko
(2023) analyze the calibration--refinement decomposition for general proper
rules. Related connections between proper losses, calibration, and decision
guarantees appear in Blasiok, Gopalan, Hu, and Nakkiran (2023) and Luo,
Senapati, and Sharan (2025).

\subsection{Outline of the paper\label{sus:intro-outline}}

Section \ref{s:setup} sets up the framework: scoring rules, divergences, and
the Brier, calibration, and refinement scores. Section \ref{s:proper}
introduces the \textquotedblleft proper\textquotedblright\ concepts (proper
calibration and proper calibeating). Section \ref{s:proper calibration}
proves that calibration automatically yields proper calibration. Section \ref%
{s:proper calibeating} presents the counterexample showing that calibeating
does not imply proper calibeating, establishes the three routes to proper
calibeating and proper-Li calibeating, and extends these results to
multicalibeating. The decision-making under uncertainty setup is discussed
in Section \ref{s:utility}, which establishes how utility maximization,
universal no regret, and universal subsuming relate to proper calibration,
calibeating, and complete calibeating. The Appendix collects background
material on scoring rules, includes deferred proofs, and provides further
analysis of the joint-binning condition and utility-improving properties.

\section{The Setup\label{s:setup}}

The setup follows our previous work (Foster and Hart 2018, 2021, 2023), with
the set of forecasts $C$ specified as a probability simplex.

Let $A$ be a finite set of \emph{states of nature} (\emph{states} for short;
also referred to as \textquotedblleft actions\textquotedblright ), and let%
\footnote{%
We write $\Delta (Z)$ for the set of probability distributions over the set $%
Z$.} $C:=\Delta (A)=\{c\in \mathbb{R}_{+}^{A}:\sum_{a\in A}c(a)=1\}$, the
simplex of probability distributions over the set $A$, be the set of \emph{%
forecasts.} We identify the elements of $A$ with the unit vectors of $C$.

\subsection{Scoring Rules\label{sus:scoring rules}}

A \emph{scoring rule}\footnote{%
See Appendix \ref{sus-a:scoring rules} for further details.} $L_{A}:A\times
C\rightarrow \mathbb{R}$ assigns a loss $L_{A}(a,c)$ to every forecast $c\in
C$ and every realized state\footnote{%
There are scoring rules, such as the logarithmic scoring rule, that allow
the loss to be infinite; as they are not bounded, we will not deal with them
here (for \textquotedblleft log-calibeating,\textquotedblright\ see Appendix
A.9 in Foster and Hart 2023f).} $a\in A$. The function $L_{A}$ is linearly
extended to $L:C\times C\rightarrow \mathbb{R}$ by\footnote{%
We thus have $L_{A}(a,c)=L(\mathbf{1}_{a},c)$.} $L(d,c):=\sum_{a\in
A}d(a)L_{A}(a,c)$; thus, $L(d,c)=\mathbb{E}_{a\sim d}\left[ L_{A}(a,c)\right]
$ is the expected loss when the state is drawn from the distribution $d\in
C. $ Letting $\mathbf{L}(c)$ denote the vector $(L_{A}(a,c))_{a\in A}$ in $%
\mathbb{R}^{A}$, we thus have%
\begin{equation}
L(d,c)=d\cdot \mathbf{L}(c)  \label{eq:d*Lambda(c)}
\end{equation}%
for every $c,d\in C.$

A scoring rule $L$ is \emph{proper} if $L(d,c)\geq L(d,d)$ for every $c,d\in
C$; i.e., forecasting the true distribution minimizes the expected loss (if $%
L(d,c)>L(d,d)$ for all $c\neq d$, then $L$ is \emph{strictly proper}). The
deviation of a forecast $c$ from the \textquotedblleft
perfect\textquotedblright\ forecast $d$ is measured by the $L$\emph{%
-divergence}\footnote{%
At times this is written $D(d\,||\,c)$ (as for the Kullback--Leibler
divergence).} $D\equiv D^{L}:C\times C\rightarrow \mathbb{R}$, defined by%
\begin{equation}
D(d,c)\,:=\,L(d,c)-L(d,d).  \label{eq:D-le}
\end{equation}%
Thus, $L$ is proper if and only if $D$ is always $\geq 0.$ The standard
scoring rule is the \emph{quadratic (Brier) scoring rule}, whose divergence
is $D(d,c)=\left\Vert c-d\right\Vert ^{2}$ (and $L_{A}(a,c)=-2c(a)+\left%
\Vert c\right\Vert ^{2}$).

We will deal here with scoring rules that are bounded or
Lipschitz-continuous.\footnote{%
One may also consider $\alpha $-H\"{o}lder continuity for $0<\alpha <1$; see
Remark (b) in Section \ref{susus:online refinement}.} To avoid superfluous
constants, we say that a scoring rule $L$ is $M$\emph{-bounded} if%
\begin{equation*}
\left\Vert \mathbf{L}(c)-\mathbf{L}(c^{\prime })\right\Vert \leq M,
\end{equation*}%
and $M$\emph{-Lipschitz} if%
\begin{equation*}
\left\Vert \mathbf{L}(c)-\mathbf{L}(c^{\prime })\right\Vert \leq M\left\Vert
c-c^{\prime }\right\Vert
\end{equation*}%
for all $c,c^{\prime }\in C$, where $\left\Vert \cdot \right\Vert $ denotes
the standard Euclidean norm and $M$ is finite.

The following proposition collects several useful properties; for details
and proofs see Appendix \ref{sus-a:scoring rules}.

\begin{proposition}
\label{p:L}Let $L$ be a proper scoring rule.

\begin{description}
\item[(i)] The function $H\equiv H^{L}:C\rightarrow \mathbb{R}$
(\textquotedblleft the $L$-entropy\textquotedblright ) given by $%
H(c):=L(c,c) $ is concave, and we have 
\begin{eqnarray*}
L(d,c) &=&H(c)-(c-d)\cdot \mathbf{L}(c)\;\;\;\text{and} \\
D(d,c) &=&H(c)-H(d)-(c-d)\cdot \mathbf{L}(c)
\end{eqnarray*}%
for every $c,d$ in $C$.

\item[(ii)] If $L$ is $M$-bounded then%
\begin{equation}
0\leq D(d,c)\leq M\left\Vert c-d\right\Vert  \label{eq:L-bounded}
\end{equation}%
for every $c,d$ in $C$.

\item[(iii)] If $L$ is $M$-Lipschitz then%
\begin{equation}
0\leq D(d,c)\leq M\left\Vert c-d\right\Vert ^{2}  \label{eq:L}
\end{equation}%
for every $c,d$ in $C$, and%
\begin{equation}
\left\vert D(d,c)-D(d,c^{\prime })\right\vert =\left\vert
L(d,c)-L(d,c^{\prime })\right\vert \leq M\left\Vert c-c^{\prime }\right\Vert
\label{eq:D-D-L}
\end{equation}%
for every $c,c^{\prime },d$ in $C$.
\end{description}
\end{proposition}

\subsection{Scores for Sequences of Forecasts\label{sus:sequences}}

We define the relevant scores for sequences of states and forecasts.
Specifically: the \emph{Brier score} is the average divergence of forecasts
from realized states; the \emph{calibration score} is the average divergence
of forecasts from the state average in all periods when that forecast is
issued; and the \emph{refinement score} is the average divergence of this
conditional state average from the realized states.

Let $t$ be the horizon; for $s=1,\ldots ,t$, let $a_{s}\in A$ be the state, $%
c_{s}\in C$ the forecast, and $i_{s}\in I$ the \textquotedblleft
bin\textquotedblright\ (for some set of bins $I$).\footnote{%
We abstract away from the specific way that the partition into bins
(\textquotedblleft binning\textquotedblright ) is determined. The \emph{%
standard} binning is by forecast: $i_{s}=c_{s}$ for all $s$.} We write $%
\mathbf{a}_{t}$ for $(a_{s})_{s=1}^{t}$ and $\mathbf{a}$ for $%
(a_{s})_{s=1}^{\infty },$ and similarly for the other sequences.\footnote{%
While for most results only $\mathbf{a}_{t},\mathbf{c}_{t},\ldots $ are
needed, we write $\mathbf{a},\mathbf{c},\ldots $ for convenience.} Assume
that the binning is a refinement of the standard binning generated by the
forecasts; i.e., all forecasts in the same bin $i$ have the same value $c$
(formally: $i_{s}=i_{r}$ implies $c_{s}=c_{r}$).\footnote{%
This allows having distinct bins with the same forecast. Formally, the
forecast is \emph{measurable} with respect to the binning.} A proper scoring
rule $L$, with corresponding divergence function $D^{L}$, generates the
following scores, which we refer to as $L$\emph{-Brier}, $L$\emph{%
-calibration}, and $L$\emph{-refinement}:\footnote{%
The $\mathcal{B}$-score depends on the forecasting sequence $\mathbf{c},$
the $\mathcal{R}$-score on the binning sequence $\mathbf{i}$, and the $%
\mathcal{K}$-score on both (of course, all of them depend on the state
sequence $\mathbf{a}$ as well). When $L$ is the standard quadratic score we
drop the superscript $L$.}%
\begin{eqnarray*}
\mathcal{B}_{t}^{L}\equiv \mathcal{B}_{t}^{L}(\mathbf{c}) &%
%TCIMACRO{\TeXButton{:=}{{\;:=\;}}}%
%BeginExpansion
{\;:=\;}%
%EndExpansion
&\frac{1}{t}\sum_{s=1}^{t}D^{L}(a_{s},c_{s}) \\
\mathcal{K}_{t}^{L}\equiv \mathcal{K}_{t}^{L}(\mathbf{c};\mathbf{i}) &%
%TCIMACRO{\TeXButton{:=}{{\;:=\;}}}%
%BeginExpansion
{\;:=\;}%
%EndExpansion
&\frac{1}{t}\sum_{s=1}^{t}D^{L}(\bar{a}_{t}(i_{s}),c_{s}) \\
\mathcal{R}_{t}^{L}\equiv \mathcal{R}_{t}^{L}(\mathbf{i}) &%
%TCIMACRO{\TeXButton{:=}{{\;:=\;}}}%
%BeginExpansion
{\;:=\;}%
%EndExpansion
&\frac{1}{t}\sum_{s=1}^{t}D^{L}(a_{s},\bar{a}_{t}(i_{s})),
\end{eqnarray*}%
where for each bin $i$ in $I$%
\begin{equation*}
n_{t}(i)\,%
%TCIMACRO{\TeXButton{:=}{{\;:=\;}}}%
%BeginExpansion
{\;:=\;}%
%EndExpansion
\,\left\vert \{s\leq t:i_{s}=i\}\right\vert
\end{equation*}%
is the number of entries in bin $i$, and, when $n_{t}(i)>0$,%
\begin{equation*}
\bar{a}_{t}(i)\,%
%TCIMACRO{\TeXButton{:=}{{\;:=\;}}}%
%BeginExpansion
{\;:=\;}%
%EndExpansion
\,\frac{1}{n_{t}(i)}\sum_{s\leq t:i_{s}=i}a_{s}
\end{equation*}%
is the state average in bin $i$. For the standard binning given by the
forecasts, i.e., when $\mathbf{i}=\mathbf{c}$, we shorten $\mathcal{K}^{L}(%
\mathbf{c};\mathbf{c})$ to $\mathcal{K}^{L}(\mathbf{c}).$

Let%
\begin{equation*}
\mathcal{H}_{t}^{L}%
%TCIMACRO{\TeXButton{:=}{{\;:=\;}}}%
%BeginExpansion
{\;:=\;}%
%EndExpansion
\frac{1}{t}\sum_{s=1}^{t}H^{L}(a_{s})=\frac{1}{t}\sum_{s=1}^{t}L(a_{s},a_{s})
\end{equation*}%
be the average $L$-entropy; by definition of $D$ we get%
\begin{equation}
\mathcal{B}_{t}^{L}(\mathbf{c})=\frac{1}{t}\sum_{s=1}^{t}L(a_{s},c_{s})-%
\mathcal{H}_{t}^{L}.  \label{eq:B=avg-L - H}
\end{equation}%
Next, summing by bins\footnote{%
Throughout, sums over bins $i$ are understood to include only nonempty bins,
i.e., bins $i$ with $n_{t}(i)>0$.} and using the linearity of $L$ in its
first argument yields%
\begin{eqnarray}
\mathcal{R}_{t}^{L}(\mathbf{i}) &=&\frac{1}{t}\sum_{s=1}^{t}L(a_{s},\bar{a}%
_{t}(i_{s}))-\mathcal{H}_{t}^{L}  \notag \\
&=&\sum_{i\in I}\left( \frac{n_{t}(i)}{t}\right) \left( \frac{1}{n_{t}(i)}%
\sum_{s\leq t:i_{s}=i}L(a_{s},\bar{a}_{t}(i))\right) -\mathcal{H}_{t}^{L} 
\notag \\
&=&\sum_{i\in I}\left( \frac{n_{t}(i)}{t}\right) L(\bar{a}_{t}(i),\bar{a}%
_{t}(i))-\mathcal{H}_{t}^{L}  \notag \\
&=&\sum_{i\in I}\left( \frac{n_{t}(i)}{t}\right) H^{L}(\bar{a}_{t}(i))-%
\mathcal{H}_{t}^{L}.  \label{eq:R-H-H}
\end{eqnarray}

The classic decomposition of the quadratic Brier score as the sum of
calibration and refinement (see the Introduction) easily generalizes to all
proper scoring rules $L$: 
\begin{equation*}
\mathcal{B}_{t}^{L}(\mathbf{c})=\mathcal{K}_{t}^{L}(\mathbf{c})+\mathcal{R}%
_{t}^{L}(\mathbf{c}).
\end{equation*}
We state this more generally, for binning sequences that may be finer than
the forecasting sequence.

\begin{theorem}
\label{th:S-decompose}Let the binning sequence $\mathbf{i}$ be a refinement
of the forecasting sequence $\mathbf{c}$; then 
\begin{equation*}
\mathcal{B}_{t}^{L}(\mathbf{c})=\mathcal{K}_{t}^{L}(\mathbf{c};\mathbf{i})+%
\mathcal{R}_{t}^{L}(\mathbf{i})
\end{equation*}%
for every proper scoring rule $L$.
\end{theorem}

\begin{proof}
Let $c^{i}$ denote the forecast in bin $i$. Summing by bins and using the
linearity of $L$ in its first argument yields 
\begin{eqnarray*}
\mathcal{B}_{t}^{L}(\mathbf{c}) &=&\sum_{i\in I}\left( \frac{n_{t}(i)}{t}%
\right) \left( \frac{1}{n_{t}(i)}\sum_{s\leq t:i_{s}=i}L(a_{s},c^{i})\right)
-\mathcal{H}_{t}^{L} \\
&=&\sum_{i\in I}\left( \frac{n_{t}(i)}{t}\right) L(\bar{a}_{t}(i),c^{i})-%
\mathcal{H}_{t}^{L},
\end{eqnarray*}%
and%
\begin{eqnarray*}
\mathcal{K}_{t}^{L}(\mathbf{c};\mathbf{i}) &=&\sum_{i}\left( \frac{n_{t}(i)}{%
t}\right) \left[ L(\bar{a}_{t}(i),c^{i})-L(\bar{a}_{t}(i),\bar{a}_{t}(i))%
\right] \\
&=&\sum_{i\in I}\left( \frac{n_{t}(i)}{t}\right) L(\bar{a}%
_{t}(i),c^{i})-\sum_{i\in I}\left( \frac{n_{t}(i)}{t}\right) H^{L}(\bar{a}%
_{t}(i)).
\end{eqnarray*}%
Subtracting gives $\mathcal{B}-\mathcal{K}=\mathcal{R}$ by (\ref{eq:R-H-H}).
\end{proof}

\medskip

As a consequence, $\mathcal{R}_{t}^{L}(\mathbf{i})$ may be viewed as the 
\emph{minimal }$L$\emph{-Brier score} subject to the binning $\mathbf{i}$,
i.e.,%
\begin{equation*}
\mathcal{R}_{t}^{L}(\mathbf{i})=\min_{\phi :I\rightarrow C}\mathcal{B}%
_{t}^{L}(\phi (\mathbf{i})),
\end{equation*}%
where $\phi (\mathbf{i})=(\phi (i_{s}))_{s\geq 1}$ (cf. (2) in Foster and
Hart 2023; see also Appendix A.10 in Foster and Hart 2023ff). This says that
among all forecasting sequences $\mathbf{c}=\phi (\mathbf{i})$ that
\textquotedblleft respect\textquotedblright\ the binning $\mathbf{i}$ (i.e.,
in all periods that are in the same bin $i$ the forecast is the same,
namely, $\phi (i)$; formally, $\mathbf{i}$ refines $\mathbf{c}$), the $L$%
-Brier score is minimal when the forecast is the state average of the bin
(i.e., $c=\phi (i)=\bar{a}_{t}(i)$). Indeed, in this case $\mathcal{K}%
_{t}^{L}=0$, and so $\mathcal{B}_{t}^{L}=\mathcal{R}_{t}^{L}$ (whereas in
general $\mathcal{B}_{t}^{L}\geq \mathcal{R}_{t}^{L}$, because $\mathcal{K}%
_{t}^{L}\geq 0$). Moreover, this minimum is attained simultaneously for 
\emph{all} proper scoring rules $L$.

\subsubsection{General Binning Sequences\label{susus:fractional}}

Following Foster and Hart (2021), we now consider general binnings for which
the allocation into bins may be fractional. Let $I$ be a finite or countably
infinite set of bins;\footnote{%
For pure binnings the number of bins is always finite (up to time $t$, it is
at most $t$).} a \emph{general binning sequence} $\mathbf{f}%
=(f_{s})_{s=1,2,\ldots }$ specifies in each period $s$ the fraction $%
f_{s}(i)\geq 0$ that is assigned to each bin $i\in I$, where $\sum_{i\in
I}f_{s}(i)=1$; thus, $f_{s}$ may be viewed as a probability distribution on $%
I$, i.e., $f_{s}\in \Delta (I)$. For example, given a fractional binning $%
\Pi =(w_{i})_{i\in I}$, where $w_{i}:C\rightarrow \lbrack 0,1]$ and $%
\sum_{i\in I}w_{i}(c)=1$ for every $c\in C$ (see Foster and Hart 2021 and
Section \ref{susus:continuous calibration} below), we put\footnote{%
The binning rule is thus time-independent: the fraction $f_{s}(i)$ depends
only on the forecast $c_{s}$ and not on the \textquotedblleft
calendar\textquotedblright\ period $s$.}$^{,}$\footnote{%
To avoid confusion, we refer to $\Pi =(w_{i})_{i\in I}$ as a
\textquotedblleft fractional binning\textquotedblright\ (as in our previous
papers), and to a sequence $\mathbf{f}=(f_{s})_{s}$ with $f_{s}\in \Delta
(I) $ as a \textquotedblleft general binning.\textquotedblright\ Thus, a
fractional binning $\Pi $ applied to a forecasting sequence $\mathbf{c}$
generates a general binning $\Pi (\mathbf{c})$.} $f_{s}(i)=w_{i}(c_{s})$.
When each $f_{s}$ is a unit vector, i.e., in each period there is a single
bin, we call the binning sequence \emph{pure}.

The definitions of calibration and refinement naturally extend to general
binnings:%
\begin{eqnarray*}
\mathcal{K}_{t}^{L}(\mathbf{c};\mathbf{f}) &%
%TCIMACRO{\TeXButton{:=}{{\;:=\;}}}%
%BeginExpansion
{\;:=\;}%
%EndExpansion
&\frac{1}{t}\sum_{i\in I}n_{t}(i)D^{L}(\bar{a}_{t}(i),\bar{c}_{t}(i)) \\
\mathcal{R}_{t}^{L}(\mathbf{f}) &%
%TCIMACRO{\TeXButton{:=}{{\;:=\;}}}%
%BeginExpansion
{\;:=\;}%
%EndExpansion
&\frac{1}{t}\sum_{i\in I}\sum_{s=1}^{t}f_{s}(i)D^{L}(a_{s},\bar{a}_{t}(i)),
\end{eqnarray*}%
where for each bin $i$ in $I$%
\begin{equation*}
n_{t}(i)\,%
%TCIMACRO{\TeXButton{:=}{{\;:=\;}}}%
%BeginExpansion
{\;:=\;}%
%EndExpansion
\,\sum_{s=1}^{t}f_{s}(i)
\end{equation*}%
is the total weight in bin $i$, and, when $n_{t}(i)>0$,%
\begin{eqnarray*}
\bar{a}_{t}(i) &%
%TCIMACRO{\TeXButton{:=}{{\;:=\;}}}%
%BeginExpansion
{\;:=\;}%
%EndExpansion
&\sum_{s=1}^{t}\left( \frac{f_{s}(i)}{n_{t}(i)}\right) a_{s} \\
\bar{c}_{t}(i) &%
%TCIMACRO{\TeXButton{:=}{{\;:=\;}}}%
%BeginExpansion
{\;:=\;}%
%EndExpansion
&\sum_{s=1}^{t}\left( \frac{f_{s}(i)}{n_{t}(i)}\right) c_{s}
\end{eqnarray*}%
are the state average and the average forecast in bin $i$.

These definitions clearly reduce to the previous definitions when the
binning sequence is pure and it refines the standard by-forecast-binning,
because then each bin contains a single forecast value. For general binning
sequences, where a bin may contain multiple forecast values, the
decomposition of Theorem \ref{th:S-decompose} no longer holds.\footnote{%
For the quadratic scoring rule there is another such decomposition, but with
a different definition of refinement, namely, as the average bin-variance of
the \emph{differences} $a_{t}-c_{t}$. Carrying this out for a general proper
scoring rule yields an average of \emph{differences between divergences}
(which is in general not a divergence). We sidestep this by using Theorem %
\ref{th:decompose-delta} below.} However, we will now show that it continues
to approximately hold for \textquotedblleft local\textquotedblright\ binning
sequences where the forecasts in each bin are close to one another and the
scoring rule is Lipschitz.

Let $\delta \geq 0$; a general binning sequence $\mathbf{f}$ is $\delta $%
\emph{-local with respect to the sequence }$\mathbf{c}$ if for each $i$
there is a closed ball $\bar{B}(y^{i};\delta )$ with center $y^{i}\in C$ and
radius $\delta $ such that $f_{s}(i)>0$ implies $\left\Vert
c_{s}-y^{i}\right\Vert \leq \delta $; i.e., all forecasts in bin $i$ lie in $%
\bar{B}(y^{i};\delta ).$ Pure binning sequences that are a refinement of the
standard by-forecast binning are thus $0$-local. The following is a
generalization of the Decomposition Theorem \ref{th:S-decompose}; for the
quadratic scoring, this is Lemma 15 in Foster and Hart (2023ff).

\begin{theorem}
\label{th:decompose-delta}Let $\delta \geq 0$; if the general binning
sequence $\mathbf{f}$ is $\delta $-local with respect to the forecasting
sequence $\mathbf{c}$, then%
\begin{equation*}
\left\vert \mathcal{B}_{t}^{L}(\mathbf{c})-\left( \mathcal{K}_{t}^{L}(%
\mathbf{c};\mathbf{f})+\mathcal{R}_{t}^{L}(\mathbf{f})\right) \right\vert
\leq 2M\delta
\end{equation*}%
for every $M$-Lipschitz proper scoring rule $L$.
\end{theorem}

The proof is given in Appendix \ref{sus-a:decompose-delta}.

\subsection{Calibration and Calibeating\label{sus:calib defs}}

We briefly recall the definitions of calibration, continuous calibration,
and calibeating; see Foster and Hart (2021, 2023) for details and
discussions.

A (stochastic) \emph{forecasting procedure} $\sigma $ is a mapping $\sigma
:\cup _{t\geq 1}(A^{t-1}\times C^{t-1})\rightarrow \Delta (C)$; i.e., to
each history $(\mathbf{a}_{t-1},\mathbf{c}_{t-1})$ of states and forecasts
before time $t$ the procedure $\sigma $ assigns a probability distribution $%
\sigma (\mathbf{a}_{t-1},\mathbf{c}_{t-1})$ on $C$, whose realization is the
forecast $c_{t}\in C$. The procedure $\sigma $ is \emph{deterministic} if
all these probability distributions are pure (i.e., the support of each $%
\sigma (\mathbf{a}_{t-1},\mathbf{c}_{t-1})$ consists of a single point $%
c_{t} $ in $C$);\footnote{%
To avoid confusion, we note that a deterministic procedure yields a single
forecast each period, but says nothing on whether that forecast is pure
(i.e., puts probability $1$ on an action $a\in A$) or mixed.} thus, $\sigma
:\cup _{t\geq 1}(A^{t-1}\times C^{t-1})\rightarrow C$. The procedure $\sigma 
$ is $\delta $\emph{-deterministic} for some $\delta >0$ if the support of
each $\sigma (\mathbf{a}_{t-1},\mathbf{c}_{t-1})$ is included in some ball
of radius $\delta $.

We will always denote the forecasting sequence of our procedure by $\mathbf{c%
}=(c_{t})_{t\geq 1}$.

\subsubsection{Calibration\label{susus:calibration}}

Let $\varepsilon \geq 0;$ a forecasting procedure $\sigma $ is $\varepsilon $%
\emph{-calibrated} (Foster and Vohra 1998) if\footnote{%
The reason for $\varepsilon ^{2}$ on the right-hand side is that the
calibration score is based on squared distances $\left\Vert \bar{a}%
-c\right\Vert ^{2}$; see footnote 7 in Foster and Hart (2023).}%
\begin{equation*}
\varlimsup_{t\rightarrow \infty }\left( \sup_{\mathbf{a}_{t}}\mathbb{E}\left[
\mathcal{K}_{t}(\mathbf{c})\right] \right) \leq \varepsilon ^{2}
\end{equation*}%
(here, and in the sequel, the expectation $\mathbb{E}$ is taken over the
random forecasts of $\sigma $).

\subsubsection{Continuous Calibration\label{susus:continuous calibration}}

A \emph{fractional binning} $\Pi =(w_{i})_{i\in I}$ is a finite or countably
infinite collection of weight functions $w_{i}:C\rightarrow \lbrack 0,1]$
such that $\sum_{i\in I}w_{i}(c)=1$ for all $c\in C$. Thus, when the
forecast is $c$, the fraction $w_{i}(c)$ goes into bin $i$, and $\Pi
(c):=(w_{i}(c))_{i\in I}\in \Delta (I)$ may be viewed as a probability
distribution over $I$. A \emph{continuous} binning $\Pi =(w_{i})_{i\in I}$
is a fractional binning where all the functions $w_{i}$ are continuous
functions on $C$. A forecasting sequence $\mathbf{c}$ generates a general
binning $\Pi (\mathbf{c})=(\Pi (c_{t}))_{t\geq 1}$ (i.e., at time $t,$ the
fraction that goes into $i$ is $\left( \Pi (c_{t})\right) _{i}=w_{i}(c_{t})$%
).

A deterministic forecasting procedure $\sigma $ is \emph{continuously
calibrated} (Foster and Hart 2021) if\footnote{%
Since continuous calibration can always be obtained by a deterministic
procedure, for which the corresponding calibration score converges to $0$,
we dispense with the expectation $\mathbb{E}$ and consider only $\varepsilon
=0$.}%
\begin{equation}
\lim_{t\rightarrow \infty }\left( \sup_{\mathbf{a}_{t}}\mathcal{K}_{t}(%
\mathbf{c};\Pi (\mathbf{c}))\right) =0  \label{eq:cont-calib}
\end{equation}%
for every continuous binning $\Pi $. Proposition 3 in Foster and Hart (2021)
and Proposition 12 in Foster and Hart (2023ff) show that it suffices to
require (\ref{eq:cont-calib}) for one specific continuous binning
(respectively, $\Pi _{0}$ and $\Pi ^{\ast }$); thus, $\sigma $ is
continuously calibrated if and only if (\ref{eq:cont-calib}) holds for $\Pi
=\Pi _{0}$ or for $\Pi =\Pi ^{\ast }$.

\subsubsection{Calibeating\label{susus:calibeating}}

Let $B$ be an arbitrary set, and $\mathbf{b}=(b_{t})_{t\geq 1}$ a sequence
of reference \textquotedblleft forecasts\textquotedblright\ $b_{t}$ in $B$.
We assume that in each period $t$ the forecast $b_{t}$ is announced before
the forecast $c_{t}$ is provided; thus, the distribution of $c_{t}$ may
depend on the past history $h_{t-1}=(\mathbf{a}_{t-1},\mathbf{c}_{t-1},%
\mathbf{b}_{t-1})$ as well as the current period's $b_{t}$. A $\mathbf{b}$%
\emph{-based forecasting procedure} $\zeta $ is a mapping\emph{\ }\linebreak 
$\zeta :\cup _{t\geq 1}(A^{t-1}\times C^{t-1}\times B^{t})\rightarrow \Delta
(C)$.

Let $\varepsilon \geq 0;$ a $\mathbf{b}$-based forecasting procedure $\zeta $
is $(\varepsilon ,B)$\emph{-calibeating} (Foster and Hart 2023) if%
\begin{equation}
\varlimsup_{t\rightarrow \infty }\left( \sup_{\mathbf{a}_{t}\in A^{t},%
\mathbf{b}_{t}\in B^{t}}\mathbb{E}\left[ \mathcal{B}_{t}(\mathbf{c})\mathbf{-%
}\mathcal{R}_{t}(\mathbf{b})\right] \right) \leq \varepsilon ^{2}.
\label{eq:calibeat-def-R1}
\end{equation}%
When $B\subseteq C$, and so the sequence $\mathbf{b}$ consists of forecasts
on $A$, calibeating yields (ignoring the error terms and the expectation) 
\begin{equation*}
\mathcal{B}_{t}(\mathbf{c})\leq \mathcal{R}_{t}(\mathbf{b})=\mathcal{B}_{t}(%
\mathbf{b})\mathbf{-}\mathcal{K}_{t}(\mathbf{b}).
\end{equation*}%
This means that the forecasting sequence $\mathbf{c}$ does not merely
achieve a lower Brier score than the reference sequence $\mathbf{b}$; it
\textquotedblleft beats\textquotedblright\ it by an amount that is at least $%
\mathbf{b}$'s own calibration score $\mathcal{K}_{t}(\mathbf{b})$ (hence our
coining of the term \textquotedblleft calibeating\textquotedblright ). As we
show in Foster and Hart (2023), the refinement score $\mathcal{R}_{t}(%
\mathbf{b})$ captures the \textquotedblleft expertise\textquotedblright\ of $%
\mathbf{b}$, measured by how effectively it partitions different time
periods into bins. Informally, $\mathbf{c}$ gains the calibration of $%
\mathbf{b}$ without sacrificing its expertise.

\subsubsection{Complete Calibeating\label{susus:complete calibeating}}

In Foster and Hart (2023, Section 6) we showed that a stronger form of
calibeating, namely, calibeating the joint binning, guarantees also that the
calibeating sequence is itself calibrated. We will see below that this
notion is important in its own right, and we will refer to it as
\textquotedblleft complete\textquotedblright\ calibeating (see Section \ref%
{sus:joint} and particularly Section \ref{susus:u-subsuming}, which
justifies the terminology: complete calibeating subsumes all the content
about states embodied in the reference sequence, and is thus the
\textquotedblleft ultimate calib-beating\textquotedblright ).

Let $\varepsilon \geq 0;$ a $\mathbf{b}$-based forecasting procedure $\zeta $
is \emph{complete }$(\varepsilon ,B)$-\emph{calibeating} if it $(\varepsilon
,B)$-calibeats the joint $\mathbf{b}\times \mathbf{c}$, i.e.,%
\begin{equation}
\varlimsup_{t\rightarrow \infty }\left( \sup_{\mathbf{a}_{t}\in A^{t},%
\mathbf{b}_{t}\in B^{t}}\mathbb{E}\left[ \mathcal{B}_{t}(\mathbf{c})\mathbf{-%
}\mathcal{R}_{t}(\mathbf{b}\times \mathbf{c})\right] \right) \leq
\varepsilon ^{2}.  \label{eq:complete-calib-def}
\end{equation}%
Since $\mathcal{R}_{t}(\mathbf{b}\times \mathbf{c})\leq \mathcal{R}_{t}(%
\mathbf{b})$ (the refinement score can only decrease when refining the
binning; see Appendix A.4 in Foster and Hart 2023, and also Proposition \ref%
{p:refine-S} below), complete $(\varepsilon ,B)$\emph{-}calibeating implies $%
(\varepsilon ,B)$\emph{-}calibeating; since $\mathcal{K}_{t}(\mathbf{c})=%
\mathcal{B}_{t}(\mathbf{c})-\mathcal{R}_{t}(\mathbf{c})\leq \mathcal{B}_{t}(%
\mathbf{c})-\mathcal{R}_{t}(\mathbf{b}\times \mathbf{c})$, it implies $%
\varepsilon $-calibration. Moreover, by the decomposition theorem we have $%
\mathcal{B}_{t}(\mathbf{c})\mathbf{-}\mathcal{R}_{t}(\mathbf{b}\times 
\mathbf{c})=\mathcal{K}_{t}(\mathbf{c};\mathbf{b}\times \mathbf{c})$, and so:

\begin{proposition}
\label{p:complete=calibrated-joint}A $\mathbf{b}$-based forecasting
procedure $\zeta $ is complete $(\varepsilon ,B)$-calibeating if and only if
it is $\varepsilon $-calibrated with respect to the joint binning $\mathbf{b}%
\times \mathbf{c}$, i.e., 
\begin{equation*}
\varlimsup_{t\rightarrow \infty }\left( \sup_{\mathbf{a}_{t}\in A^{t},%
\mathbf{b}_{t}\in B^{t}}\mathbb{E}\left[ \mathcal{K}_{t}(\mathbf{c};\mathbf{b%
}\times \mathbf{c})\right] \right) \leq \varepsilon ^{2}.
\end{equation*}
\end{proposition}

While $\mathbf{c}$ being calibrated means (ignoring error terms) that the
state average $\bar{a}(c)$ of each $c$-bin is equal to $c$, being calibrated
on the joint binning $\mathbf{b}\times \mathbf{c}$ means that the state
average $\bar{a}(b,c)$ in each $(b,c)$-subbin of the $c$-bin is also equal
to $c$. This is exactly what complete calibeating entails.

\section{\textquotedblleft Proper\textquotedblright\ Concepts\label{s:proper}%
}

We define a procedure as \textquotedblleft proper\textquotedblright
-calibrated if it is calibrated with respect to every bounded proper scoring
rule $L$; that is, its $L$-calibration score converges to zero as the
horizon increases. Moreover, we require uniformity in the scoring rule $L$.
Since multiplying a scoring rule by $\lambda >0$ multiplies all scores by $%
\lambda $, achieving uniform convergence requires normalizing the scoring
rules. A convenient normalization is to divide by the bounding constant and
so obtain $1$-bounded scoring rules; for Lipschitz proper scoring rules, we
divide by the Lipschitz constant to obtain $1$-Lipschitz scoring rules.
Proper calibeating will be defined similarly.

Let $\mathcal{L}$ denote the class of all bounded proper scoring rules, and $%
\mathcal{L}_{1}$ the subclass of $1$-bounded proper scoring rules; let $%
\mathcal{L}^{\mathrm{Li}}$ denote the class of all Lipschitz proper scoring
rules, and $\mathcal{L}_{1}^{\mathrm{Li}}$ the subclass of $1$-Lipschitz
proper scoring rules. We will say that a procedure is \emph{(uniformly)
proper calibrated/proper calibeating} if the corresponding guarantee holds
simultaneously for all scoring rules in $\mathcal{L}_{1}$, and \emph{%
(uniformly) proper-}Li\emph{-calibrated/proper-}Li\emph{-calibeating} if it
holds simultaneously for all scoring rules in $\mathcal{L}_{1}^{\mathrm{Li}}$
(for brevity we will usually drop the term \textquotedblleft
uniform\textquotedblright ). The formal definitions are as follows, for $%
\varepsilon \geq 0$ (when\footnote{%
An asymptotic $\varepsilon =0$ error can always be achieved by the so-called
\textquotedblleft doubling trick,\textquotedblright\ whereby one lowers $%
\varepsilon $ with time; see Cesa-Bianchi and Lugosi (2006).} $\varepsilon
=0 $ we say \textquotedblleft proper calibrated\textquotedblright\ instead
of \textquotedblleft proper $0$-calibrated,\textquotedblright\ and so on):

\begin{itemize}
\item A forecasting procedure is \emph{(uniformly) proper }$\varepsilon $-%
\emph{calibrated} (following Foster and Vohra 1998) if%
\begin{equation*}
\varlimsup_{t\rightarrow \infty }\left( \sup_{L\in \mathcal{L}_{1}}\sup_{%
\mathbf{a}_{t}\in A^{t}}\mathbb{E}\left[ \mathcal{K}_{t}^{L}(\mathbf{c})%
\right] \right) \leq \varepsilon ^{2},
\end{equation*}%
and is \emph{(uniformly) proper-Lipschitz }$\varepsilon $-\emph{calibrated},
or \emph{proper-Li }$\varepsilon $-\emph{calibrated}, if%
\begin{equation*}
\varlimsup_{t\rightarrow \infty }\left( \sup_{L\in \mathcal{L}_{1}^{\mathrm{%
Li}}}\sup_{\mathbf{a}_{t}\in A^{t}}\mathbb{E}\left[ \mathcal{K}_{t}^{L}(%
\mathbf{c})\right] \right) \leq \varepsilon ^{2}.
\end{equation*}

\item A deterministic forecasting procedure $\sigma $ is \emph{(uniformly)
proper continuously-calibrated} (following Foster and Hart 2021) if%
\begin{equation*}
\lim_{t\rightarrow \infty }\left( \sup_{L\in \mathcal{L}_{1}}\sup_{\mathbf{a}%
_{t}\in A^{t}}\mathcal{K}_{t}^{L}(\mathbf{c};\Pi (\mathbf{c}))\right) =0
\end{equation*}%
for every continuous binning $\Pi $, and is \emph{(uniformly)
proper-Lipschitz continuously-calibrated}, or \emph{proper-Li
continuously-calibrated}, if%
\begin{equation*}
\lim_{t\rightarrow \infty }\left( \sup_{L\in \mathcal{L}_{1}^{\mathrm{Li}%
}}\sup_{\mathbf{a}_{t}\in A^{t}}\mathcal{K}_{t}^{L}(\mathbf{c};\Pi (\mathbf{c%
}))\right) =0
\end{equation*}%
for every continuous binning $\Pi $.

\item Let $B$ be a finite set; a $\mathbf{b}$-based forecasting procedure $%
\sigma $ is \emph{(uniformly) proper }$(\varepsilon ,B)$-\emph{calibeating }%
(following Foster and Hart 2023) if\footnote{%
Here the term $\mathcal{R}_{t}^{L}(\mathbf{b})$ may be taken out of the
expectation (which is over the randomizations of $\sigma $).}%
\begin{equation*}
\varlimsup_{t\rightarrow \infty }\left( \sup_{L\in \mathcal{L}_{1}}\sup_{%
\mathbf{a}_{t}\in A^{t},\mathbf{b}_{t}\in B^{t}}\mathbb{E}\left[ \mathcal{B}%
_{t}^{L}(\mathbf{c})-\mathcal{R}_{t}^{L}(\mathbf{b})\right] \right) \leq
\varepsilon ^{2};
\end{equation*}%
is \emph{(uniformly) proper complete }$(\varepsilon ,B)$-\emph{calibeating }%
if%
\begin{equation*}
\varlimsup_{t\rightarrow \infty }\left( \sup_{L\in \mathcal{L}_{1}}\sup_{%
\mathbf{a}_{t}\in A^{t},\mathbf{b}_{t}\in B^{t}}\mathbb{E}\left[ \mathcal{B}%
_{t}^{L}(\mathbf{c})-\mathcal{R}_{t}^{L}(\mathbf{b}\times \mathbf{c})\right]
\right) \leq \varepsilon ^{2};
\end{equation*}%
and is \emph{(uniformly) proper-Lipschitz }$(\varepsilon ,B)$\emph{%
-calibeating}, or \emph{proper-Li }$(\varepsilon ,B)$-\emph{calibeating}, if%
\begin{equation*}
\varlimsup_{t\rightarrow \infty }\left( \sup_{L\in \mathcal{L}_{1}^{\mathrm{%
Li}}}\sup_{\mathbf{a}_{t}\in A^{t},\mathbf{b}_{t}\in B^{t}}\mathbb{E}\left[ 
\mathcal{B}_{t}^{L}(\mathbf{c})-\mathcal{R}_{t}^{L}(\mathbf{b})\right]
\right) \leq \varepsilon ^{2}.
\end{equation*}
\end{itemize}

For proper complete calibeating, one may replace $\mathbb{E}\left[ \mathcal{B%
}_{t}^{L}(\mathbf{c})-\mathcal{R}_{t}^{L}(\mathbf{b}\times \mathbf{c})\right]
$ with\linebreak\ $\mathbb{E}\left[ \mathcal{K}_{t}^{L}(\mathbf{c};\mathbf{b}%
\times \mathbf{c})\right] $ (by the Decomposition Theorem \ref%
{th:S-decompose}; cf. Proposition \ref{p:complete=calibrated-joint} for the
quadratic scoring). The simple terms \emph{calibration} and \emph{calibeating%
} will from now on refer to these notions with respect to the quadratic
scoring rule\ only.

Proper $\varepsilon $-calibration thus implies that $\varlimsup_{t%
\rightarrow \infty }\sup_{\mathbf{a}_{t}}\mathbb{E}\left[ \mathcal{K}%
_{t}^{L}(\mathbf{c})\right] \leq M\varepsilon ^{2}$ for every $M$-bounded
proper scoring rule $L$ (and similarly for the other concepts).

\section{Proper Calibration\label{s:proper calibration}}

We show that standard calibration always implies proper calibration.

\begin{theorem}
\label{th:proper-calibration}If a procedure is $\varepsilon $-calibrated
then it is proper $\sqrt{\varepsilon }$-calibrated and \emph{proper}-\emph{%
Li }$\varepsilon $-calibrated, and if it is continuously calibrated then it
is proper continuously-calibrated.
\end{theorem}

\medskip

Thus, an $\varepsilon $-calibrated procedure guarantees $\varlimsup_{t%
\rightarrow \infty }\sup_{\mathbf{a}_{t}}\mathbb{E}\left[ \mathcal{K}%
_{t}^{L}(\mathbf{c})\right] \leq M\varepsilon $ for every $M$-bounded proper
scoring rule $L$, and $\varlimsup_{t\rightarrow \infty }\sup_{\mathbf{a}_{t}}%
\mathbb{E}\left[ \mathcal{K}_{t}^{L}(\mathbf{c})\right] \leq M\varepsilon
^{2}$ for every $M$-Lipschitz proper scoring rule $L$. The theorem is an
immediate consequence of the following:

\begin{proposition}
\label{p:K-KS}Let $\mathbf{c}$ be a forecasting sequence, $\mathbf{f}$ a
general binning sequence, and $L$ a proper scoring rule. If $L$ is $M$%
-bounded then%
\begin{equation*}
\mathcal{K}_{t}^{L}(\mathbf{c};\mathbf{f})\leq M\sqrt{\mathcal{K}_{t}(%
\mathbf{c};\mathbf{f})},
\end{equation*}%
and if $L$ is $M$-Lipschitz then%
\begin{equation*}
\mathcal{K}_{t}^{L}(\mathbf{c};\mathbf{f})\leq M\,\mathcal{K}_{t}(\mathbf{c};%
\mathbf{f}).
\end{equation*}
\end{proposition}

\begin{proof}
For an $M$-bounded scoring rule $L$, Proposition \ref{p:L}(ii) yields 
\begin{eqnarray*}
\mathcal{K}_{t}^{L}(\mathbf{c};\mathbf{f}) &=&\sum_{i\in I}\left( \frac{%
n_{t}(i)}{t}\right) D^{L}(\bar{a}_{t}(i),\bar{c}_{t}(i))\leq \sum_{i\in
I}\left( \frac{n_{t}(i)}{t}\right) M\left\Vert \bar{a}_{t}(i)-\bar{c}%
_{t}(i)\right\Vert \\
&=&M\sum_{i\in I}\sqrt{\frac{n_{t}(i)}{t}}\left( \sqrt{\frac{n_{t}(i)}{t}}%
\left\Vert \bar{a}_{t}(i)-\bar{c}_{t}(i)\right\Vert \right) \\
&\leq &M\left( \sum_{i\in I}\frac{n_{t}(i)}{t}\right) ^{1/2}\left(
\sum_{i\in I}\frac{n_{t}(i)}{t}\left\Vert \bar{a}_{t}(i)-\bar{c}%
_{t}(i)\right\Vert ^{2}\right) ^{1/2}=M\sqrt{\mathcal{K}_{t}(\mathbf{c};%
\mathbf{f})}
\end{eqnarray*}%
(we have used the Cauchy--Schwarz inequality and $\sum_{i}n_{t}(i)/t=1$).

For an $M$-Lipschitz scoring rule $L$, Proposition \ref{p:L}(iii) yields%
\begin{eqnarray*}
\mathcal{K}_{t}^{L}(\mathbf{c};\mathbf{f}) &=&\sum_{i\in I}\left( \frac{%
n_{t}(i)}{t}\right) D^{L}(\bar{a}_{t}(i),\bar{c}_{t}(i))\leq \sum_{i\in
I}\left( \frac{n_{t}(i)}{t}\right) M\left\Vert \bar{a}_{t}(i)-\bar{c}%
_{t}(i)\right\Vert ^{2} \\
&=&M\,\mathcal{K}_{t}(\mathbf{c};\mathbf{f}).
\end{eqnarray*}
\end{proof}

\medskip

\begin{proof}[Proof of Theorem \protect\ref{th:proper-calibration}]
By Proposition \ref{p:K-KS}: for every $1$-bounded proper scoring rule $L$
we have $\mathbb{E}\left[ \mathcal{K}_{t}^{L}(\mathbf{c})\right] \leq 
\mathbb{E}\left[ \sqrt{\mathcal{K}_{t}(\mathbf{c})}\right] \leq \sqrt{%
\mathbb{E}\left[ \mathcal{K}_{t}(\mathbf{c})\right] }$, and $\mathcal{K}%
_{t}^{L}(\mathbf{c};\Pi (\mathbf{c}))\leq \sqrt{\mathcal{K}_{t}(\mathbf{c}%
;\Pi (\mathbf{c}))}$ for every continuous binning $\Pi $; and for every $1$%
-Lipschitz proper scoring rule $L$ we have $\mathcal{K}_{t}^{L}(\mathbf{c}%
)\leq \mathcal{K}_{t}(\mathbf{c})$.
\end{proof}

\medskip

For $\varepsilon =0$ we thus have%
\begin{equation*}
\begin{tabular}{|c|}
\hline
$\text{calibration\ }\Longleftrightarrow \;\text{proper calibration}$ \\ 
\hline
\end{tabular}%
\end{equation*}

The existing results in the literature thus yield stochastic procedures that
are proper $\varepsilon $-calibrated, and deterministic procedures that are
proper continuously-calibrated. For instance, from Theorem 4 of Foster and
Hart (2023) (with $C=\Delta (A)$, and thus $\gamma ^{2}=\max_{c,c^{\prime
}\in C}\left\Vert c-c^{\prime }\right\Vert ^{2}=2$; see also Theorem 11 (S)
of Foster and Hart 2021) we get:

\begin{theorem}
\label{th:calibration S}Let $\delta >0$ and let $C_{\delta }\subset C$ be a
finite $\delta $-grid of $C$. Then there exists a stochastic $C_{\delta }$%
-forecasting procedure $\sigma $ that is proper $\sqrt{\delta }$-calibrated;
specifically, 
\begin{equation*}
\mathbb{E}\left[ \mathcal{K}_{t}^{L}(\mathbf{c})\right] \leq \left( \delta
^{2}+2|C_{\delta }|\frac{\ln t+1}{t}\right) ^{1/2}
\end{equation*}%
for all $t\geq 1$, all sequences $\mathbf{a}_{t}\in A^{t}$, and all $1$%
-bounded proper scoring rules $L$ (i.e., $L\in \mathcal{L}_{1}$). Moreover, $%
\sigma $ may be taken to be $\delta $-almost deterministic (i.e., all
randomizations are $\delta $-local).
\end{theorem}

For Lipschitz proper scoring rules we get proper-Li $\delta $-calibration;
i.e.,%
\begin{equation*}
\mathbb{E}\left[ \mathcal{K}_{t}^{L}(\mathbf{c})\right] \leq \delta
^{2}+2|C_{\delta }|\frac{\ln t+1}{t}
\end{equation*}%
for every $L$ in $\mathcal{L}_{1}^{\mathrm{Li}}$.

Next, from Theorem 11 (D) of Foster and Hart (2021) (see also Theorems 6 and
12 of Foster and Hart 2023, 2023ff) we get:

\begin{theorem}
\label{th:cont}There exists a \emph{deterministic} forecasting procedure $%
\sigma $ that is proper continuously-calibrated.
\end{theorem}

\medskip

\noindent \textbf{Remark. }In Foster and Hart (2021, 2023) we have
emphasized the important distinction between procedures of \emph{type MM}
(minmax) and procedures of \emph{type FP} (fixed point). To determine the
forecast in each period, the former requires solving a finite minmax problem
(equivalently, a finite linear programming problem), whereas the latter
requires solving a continuous fixed-point problem. The stochastic procedures
in this paper are all of type MM, whereas the deterministic and $\delta $%
-deterministic procedures (except for the \textquotedblleft simple way to
calibeat\textquotedblright\ procedure of Theorem \ref%
{th:simple-calibeat-univ} below) are of type FP.

\section{Proper Calibeating\label{s:proper calibeating}}

Unlike calibration, proper calibeating is \emph{not} a consequence of
calibeating. We show this in Section \ref{sus:example} below, and then we
exhibit three methods of achieving proper calibeating and proper-Li
calibeating. First, we prove in Section \ref{sus:joint} that complete
calibeating (i.e., calibeating the joint binning) yields proper complete
calibeating. Second, we prove in Section \ref{sus:simple} that the simple
calibeating procedure of Theorem 3 of Foster and Hart (2023) is proper-Li
calibeating (but not proper calibeating; i.e., there are bounded but
non-Lipschitz proper scoring rules for which calibeating fails). Third, we
provide in Section \ref{sus:continuous} a deterministic proper-Li
calibeating procedure that is proper continuously-calibrated.

\subsection{Calibeating Does Not Imply Proper Calibeating\label{sus:example}}

The following example shows that in general calibeating with respect to the
standard quadratic scoring rule does \emph{not} yield calibeating with
respect to other (bounded) proper scoring rules (this stands in contrast to
calibration, which, as shown above, always entails proper calibration).

\begin{example}
\label{ex:not-S-beat}In the one-dimensional case, where $A=\{0,1\}$,
consider $t=10$ periods where the states $a_{t}$ and the forecasts $b_{t}$
and $c_{t}$ (given in the table below as the forecasted probability of $a=1$%
) are as follows:%
\begin{equation*}
\begin{tabular}{c|cccccccccc}
$t$ & $1$ & $2$ & $3$ & $4$ & $5$ & $6$ & $7$ & $8$ & $9$ & $10$ \\ 
\hline\hline
$a_{t}$ & $1$ & $0$ & $0$ & $0$ & $0$ & $1$ & $1$ & $1$ & $1$ & $0$ \\ 
$b_{t}$ & $\frac{1}{5}$ & $\frac{1}{5}$ & $\frac{1}{5}$ & $\frac{1}{5}$ & $%
\frac{1}{5}$ & $\frac{4}{5}$ & $\frac{4}{5}$ & $\frac{4}{5}$ & $\frac{4}{5}$
& $\frac{4}{5}$ \\ 
$c_{t}$ & $1$ & $0$ & $\frac{1}{2}$ & $\frac{1}{2}$ & $\frac{1}{2}$ & $\frac{%
1}{2}$ & $\frac{1}{2}$ & $\frac{1}{2}$ & $1$ & $0$%
\end{tabular}%
.
\end{equation*}%
The sequences $\mathbf{b}_{t}$ and $\mathbf{c}_{t}$ are both perfectly
calibrated (i.e., $\bar{a}_{t}(b)=b$ and $\bar{a}_{t}(c)=c$ for each
forecast used), and so, for every scoring rule $L$ we have $\mathcal{K}%
_{t}^{L}(\mathbf{b})=\mathcal{K}_{t}^{L}(\mathbf{c})=0$ and\footnote{%
We use formula (\ref{eq:R-H-H}), slightly abusing notation and writing $%
H^{L}(p)$ instead of $H^{L}((p,1-p))$.} 
\begin{eqnarray*}
\mathcal{B}_{t}^{L}(\mathbf{b}) &=&\mathcal{R}_{t}^{L}(\mathbf{b})=\frac{5}{%
10}H^{L}\left( \frac{1}{5}\right) +\frac{5}{10}H^{L}\left( \frac{4}{5}%
\right) -\mathcal{H}_{t}^{L}\text{\ \ \ and} \\
\mathcal{B}_{t}^{L}(\mathbf{c}) &=&\mathcal{R}_{t}^{L}(\mathbf{c})=\frac{6}{%
10}H^{L}\left( \frac{1}{2}\right) +\frac{2}{10}H^{L}\left( 1\right) +\frac{2%
}{10}H^{L}\left( 0\right) -\mathcal{H}_{t}^{L},
\end{eqnarray*}%
where $\mathcal{H}_{t}^{L}=(5/10)H^{L}(0)+(5/10)H^{L}(1).$ For the standard
quadratic scoring rule, for which $H(p)=-(p^{2}+(1-p)^{2})$ (see Appendix %
\ref{susus-a:scoring examples}) this yields%
\begin{equation}
\mathcal{B}_{t}(\mathbf{c})=\frac{3}{10}<\frac{8}{25}=\mathcal{R}_{t}(%
\mathbf{b}),  \label{eq:ex-quad}
\end{equation}%
and so $\mathbf{c}_{t}$ calibeats $\mathbf{b}_{t}$. For the $\alpha $%
-spherical scoring rule $L$ with $\alpha =2$ (which is a bounded and
Lipschitz proper scoring rule), for which $H^{L}(p)=-(p^{2}+(1-p)^{2})^{1/2}$%
, this yields\footnote{%
The inequality $\mathcal{B}_{t}^{L}(\mathbf{c})>\mathcal{R}_{t}^{L}(\mathbf{b%
})$ holds for every $\alpha $-spherical $L$ with $\alpha \geq 2$; as $\alpha
\rightarrow \infty ,$ we get $\mathcal{B}_{t}^{L}(\mathbf{c})\rightarrow
3/10 $ and $\mathcal{B}_{t}^{L}(\mathbf{b})\rightarrow 1/5$.}%
\begin{equation}
\mathcal{B}_{t}^{L}(\mathbf{c})\approx 0.1757>0.1754\approx \mathcal{R}%
_{t}^{L}(\mathbf{b}),  \label{eq:ex-L}
\end{equation}%
and so $\mathbf{c}_{t}$ does \emph{not} $L$-calibeat $\mathbf{b}_{t}$\textbf{%
.} Repeating this sequence of length $10$ periodically yields the
inequalities (\ref{eq:ex-quad}) and (\ref{eq:ex-L}) for every $t$ that is a
multiple of $10$, and thus also in the limit as\footnote{%
Because all the scores at $t=10m+r,$ where $1\leq r\leq 9$, differ from
those at $t^{\prime }=10m$ by $O(r/t)\rightarrow 0$ as $t\rightarrow \infty $%
.} $t\rightarrow \infty $, which shows that $\mathbf{c}$ calibeats $\mathbf{b%
}$ with respect to the quadratic scoring rule but \emph{not} with respect to
the $2$-spherical scoring rule.\footnote{%
While this is demonstrated for specific sequences $\mathbf{a},\mathbf{b},%
\mathbf{c}$, it implies that any $B$-calibeating $\mathbf{b}$-based
procedure that produces the forecasting sequence $\mathbf{c}$ when the
history follows $\mathbf{a}$ and $\mathbf{b}$ is \emph{not} proper
calibeating.}
\end{example}

\subsection{Proper Complete Calibeating\label{sus:joint}}

We now consider complete calibeating, which means calibeating the joint
binning, and implies calibeating by a calibrated procedure (see Theorem 5 of
Foster and Hart 2023). We show that complete calibeating is proper, and thus
implies proper calibeating by a proper calibrated procedure.

\begin{theorem}
\label{th:proper complete calibeating}If a procedure is complete $%
(\varepsilon ,B)$-calibeating then it is proper complete $(\sqrt{\varepsilon 
},B)$-calibeating, and thus proper $\sqrt{\varepsilon }$-calibrated and
proper $(\sqrt{\varepsilon },B)$-calibeating.
\end{theorem}

\begin{proof}
Recall that $\mathcal{B}_{t}^{L}(\mathbf{c})-\mathcal{R}_{t}^{L}(\mathbf{b}%
\times \mathbf{c})=\mathcal{K}_{t}^{L}(\mathbf{c};\mathbf{b}\times \mathbf{c}%
)$ by the Decomposition Theorem \ref{th:S-decompose}, and use Proposition %
\ref{p:K-KS}: 
\begin{equation*}
\mathbb{E}\left[ \mathcal{K}_{t}^{L}(\mathbf{c};\mathbf{b}\times \mathbf{c})%
\right] \leq \mathbb{E}\left[ \sqrt{\mathcal{K}_{t}(\mathbf{c};\mathbf{b}%
\times \mathbf{c})}\right] \leq \sqrt{\mathbb{E}\left[ \mathcal{K}_{t}(%
\mathbf{c};\mathbf{b}\times \mathbf{c})\right] }
\end{equation*}%
for every $L\in \mathcal{L}_{1}$. From proper complete calibeating to proper
calibration use $\mathcal{K}_{t}^{L}(\mathbf{c})\leq \mathcal{K}_{t}^{L}(%
\mathbf{c};\mathbf{b}\times \mathbf{c})$ (by Corollary \ref{c:refine-S}
below), and to proper calibeating use $\mathcal{R}_{t}^{L}(\mathbf{c})\geq 
\mathcal{R}_{t}^{L}(\mathbf{b}\times \mathbf{c})$ (by Proposition \ref%
{p:refine-S} below).
\end{proof}

\medskip

For $\varepsilon =0$ we thus have%
\begin{equation*}
\begin{tabular}{|c|}
\hline
$\text{complete calibeating\ }\Longleftrightarrow \;\text{proper complete
calibeating}$ \\ \hline
\end{tabular}%
\end{equation*}

The existence of proper complete $\varepsilon $-calibeating procedures is
now immediate from Theorem 5 of Foster and Hart (2023).

\begin{theorem}
\label{th:beat-by-calib S}Let $B$ be a finite set, and let $C_{\delta
}\subset C$ be a finite $\delta $-grid of $C$ for some $\delta >0$. Then
there exists a stochastic $\mathbf{b}$-based $C_{\delta }$-forecasting
procedure $\zeta $ that is proper complete $(\sqrt{\delta },B)$-calibeating,
i.e.,\footnote{%
The proof below shows that the right-hand side in each formula is%
\begin{equation*}
\left( \delta ^{2}+2|B|\,|C_{\delta }|\frac{\ln t+1}{t}\right) ^{1/2}=\delta
+O\left( \sqrt{\frac{\ln t}{t}}\right) .
\end{equation*}%
}%
\begin{equation*}
\mathbb{E}\left[ \mathcal{B}_{t}^{L}(\mathbf{c})-\mathcal{R}_{t}^{L}(\mathbf{%
b}\times \mathbf{c})\right] \leq \delta +o(1),
\end{equation*}%
and thus proper $(\sqrt{\delta },B)$-calibeating and proper $\sqrt{\delta }$%
-calibrated:%
\begin{eqnarray*}
\mathbb{E}\left[ \mathcal{B}_{t}^{L}(\mathbf{c})-\mathcal{R}_{t}^{L}(\mathbf{%
b})\right] &\leq &\delta +o(1)\text{\ \ and} \\
\mathbb{E}\left[ \mathcal{K}_{t}^{L}(\mathbf{c})\right] =\mathbb{E}\left[ 
\mathcal{B}_{t}^{L}(\mathbf{c})-\mathcal{R}_{t}^{L}(\mathbf{c})\right] &\leq
&\delta +o(1).
\end{eqnarray*}%
All these inequalities hold for all $t\geq 1$ uniformly over all sequences $%
\mathbf{a}_{t}\in A^{t}$ and $\mathbf{b}_{t}\in B^{t}$, and all $1$-bounded
proper scoring rules $L$ (i.e., $L\in \mathcal{L}_{1}$). Moreover, $\zeta $
may be taken to be $\delta $-almost deterministic.
\end{theorem}

\begin{proof}
Theorem 5 of Foster and Hart (2023) with $C=\Delta (A)$ (for which we have $%
\gamma ^{2}=\max_{c,d\in C}\left\Vert c-d\right\Vert ^{2}=2$) yields%
\begin{equation*}
\mathbb{E}\left[ \mathcal{B}_{t}(\mathbf{c})-\mathcal{R}_{t}(\mathbf{b}%
\times \mathbf{c})\right] \leq \delta ^{2}+2|B|\,|C_{\delta }|\frac{\ln t+1}{%
t};
\end{equation*}%
apply Theorem \ref{th:proper complete calibeating}.
\end{proof}

\medskip

\noindent \textbf{Remarks. }\emph{(a) }The following four statements are all
equivalent to \textquotedblleft $\mathbf{c}$ completely calibeats $\mathbf{b}
$\textquotedblright :\footnote{%
For clarity we consider the simplified statements that drop expectations,
limits, and errors.}

\begin{description}
\item[(J1)] $\mathbf{c}$ calibeats $\mathbf{b}\times \mathbf{c}$, i.e., $%
\mathcal{B}(\mathbf{c})\leq \mathcal{R}(\mathbf{b}\times \mathbf{c})$;
equivalently,\footnote{%
The equivalence obtains because we always have $\mathcal{R}(\mathbf{b}\times 
\mathbf{c})\leq \mathcal{R}(\mathbf{c})\leq \mathcal{B}(\mathbf{c})$ (the
first inequality since the $\mathbf{b}\times \mathbf{c}$-binning is a
refinement of the $\mathbf{c}$-binning). Similarly for every scoring rule $L$%
, i.e., for (J4) below.} $\mathcal{B}(\mathbf{c})=\mathcal{R}(\mathbf{c})=%
\mathcal{R}(\mathbf{b}\times \mathbf{c})$.

\item[(J2)] $\mathbf{c}$ is calibrated on the $\mathbf{b\times c}$-binning,
i.e., $\mathcal{K}(\mathbf{c};\mathbf{b}\times \mathbf{c})=0$.

\item[(J3)] $\mathbf{c}$ is proper calibrated on the $\mathbf{b\times c}$%
-binning, i.e., $\mathcal{K}^{L}(\mathbf{c};\mathbf{b}\times \mathbf{c})=0$
for every $L$ in $\mathcal{L}$.

\item[(J4)] $\mathbf{c}$ properly calibeats $\mathbf{b\times c}$, i.e., $%
\mathcal{B}^{L}(\mathbf{c})\leq \mathcal{R}^{L}(\mathbf{b}\times \mathbf{c})$
for every $L$ in $\mathcal{L}$; equivalently, $\mathcal{B}^{L}(\mathbf{c})=%
\mathcal{R}^{L}(\mathbf{c})=\mathcal{R}^{L}(\mathbf{b}\times \mathbf{c})$
for every $L$ in $\mathcal{L}$.
\end{description}

\noindent Indeed, since the $\mathbf{b}\times \mathbf{c}$-binning is a
refinement of the $\mathbf{c}$-binning, we get the decomposition $\mathcal{B}%
^{L}(\mathbf{c})=\mathcal{R}^{L}(\mathbf{b}\times \mathbf{c})+\mathcal{K}%
^{L}(\mathbf{c};\mathbf{b}\times \mathbf{c})$, which immediately yields (J1) 
$\Longleftrightarrow $ (J2), and (J3) $\Longleftrightarrow $ (J4). As for
(J2) $\Longleftrightarrow $ (J3), it follows from $\mathcal{K}^{L}\leq M\,%
\sqrt{\mathcal{K}}$ for every $L$ in $\mathcal{L}$, and the fact that the
standard quadratic scoring rule is in $\mathcal{L}$.

Moreover, (J1) implies that $\mathbf{c}$ is calibrated (because $\mathcal{K}(%
\mathbf{c})=\mathcal{B}(\mathbf{c})-\mathcal{R}(\mathbf{c})=0$), and (J4)
that it is proper calibrated (because $\mathcal{K}^{L}(\mathbf{c})=\mathcal{B%
}^{L}(\mathbf{c})-\mathcal{R}^{L}(\mathbf{c})=0$).

While $\mathbf{c}$ being calibrated means that the state average $\bar{a}(c)$
of each $c$-bin is equal to $c$, being calibrated on the joint binning $%
\mathbf{b}\times \mathbf{c}$ (condition (J2)) means that the state average $%
\bar{a}(b,c)$ in each $(b,c)$-subbin of the $c$-bin is also equal to $c$.

\emph{(b) }One can achieve complete calibeating---and thus proper complete
calibeating---by running a calibration procedure for each $b\in B$
separately (i.e., for each $b\in B$ consider only the periods where $b_{s}=b$
and calibrate the forecasts there); this is exactly what the procedure of
Foster and Hart (2023, Theorem 5) does.

\emph{(c)} Calibeating together with calibration does \emph{not} suffice to
achieve proper calibeating; complete calibeating does. See Example \ref%
{ex:not-S-beat}, where $L$ is the $2$-spherical proper scoring rule: $%
\mathbf{c}$ is calibrated (and thus $L$-calibrated) and calibeats $\mathbf{b}
$, but it does not $L$-calibeat $\mathbf{b}$. Indeed, $\mathbf{c}$ is not
calibrated with respect to the $\mathbf{b}\times \mathbf{c}$-binning: the $%
(b=1/5,c=1/2)$-bin is not $\mathbf{c}$-calibrated: the state average there
is $0$ rather than $1/2$. As seen in the above proof, to obtain $L$%
-calibeating from calibeating we rely on the Decomposition Theorem \ref%
{th:S-decompose}, which requires the binning to be a refinement of the $%
\mathbf{c}$-binning; therefore, it does not apply to an arbitrary $\mathbf{b}
$-binning,\footnote{%
For our simple calibeating procedure in Section \ref{sus:simple} below we
use a different tool, namely, Proposition \ref{p:online-R-S}.} but it does
apply to the joint $\mathbf{b\times c}$-binning. Appendix \ref{sus-a:bxc}
provides further evidence on this matter: we consider a setup where we vary
only the frequencies of the bins, suggesting that a \textquotedblleft
natural\textquotedblright\ proof of proper calibeating from calibeating
together with calibration might well require complete calibeating.

\subsubsection{Refined Refinement\label{susus:refined refinement}}

For the above proof we need to generalize to proper scoring rules the result
that the refinement score can only decrease when the binning becomes finer
(for the quadratic scoring rule, this is Proposition 11 of Foster and Hart
2023). While we only need it for pure binnings in this section, we state it
for general binnings (which will be used in the following section).

Let $\mathbf{f}$ be a general binning sequence on a set of bins $I,$ and $%
\mathbf{g}$ a general binning sequence on a set of bins $J$ (thus $f_{s}\in
\Delta (I)$ and $g_{s}\in \Delta (J)$). We say that $\mathbf{f}$ is a \emph{%
refinement} of $\mathbf{g}$ (or $\mathbf{g}$ is a \emph{coarsening} of $%
\mathbf{f}$) if each $j$-bin is a union of $i$-bins, with $g_{s}(j)$ the sum
of the corresponding $f_{s}(i)$; i.e., there is a partition $I=\cup _{j\in
J}I(j)$ of $I$ into disjoint sets $I(j)$ for $j\in J,$ and $%
g_{s}(j)=\sum_{i\in I(j)}f_{s}(i)$ for every $j\in J$ and $s\geq 1$.

\begin{proposition}
\label{p:refine-S}If the general binning sequence $\mathbf{f}$ is a
refinement of the general binning sequence $\mathbf{g}$, then%
\begin{equation*}
\mathcal{R}_{t}^{L}(\mathbf{f})\leq \mathcal{R}_{t}^{L}(\mathbf{g})
\end{equation*}%
for every proper scoring rule $L$ and every $t\geq 1.$
\end{proposition}

The proof is given in Appendix \ref{sus-a:refined-refinement}.

\medskip

For pure (non-fractional) binning sequences, we obtain\footnote{%
See Footnote \ref{ftn:K(f) ge K(g)} below.}

\begin{corollary}
\label{c:refine-S}If the binning sequence $\mathbf{i}$ is a refinement of
the binning sequence $\mathbf{j}$, and $\mathbf{j}$ is a refinement of the
binning sequence $\mathbf{c}$, then%
\begin{equation*}
\mathcal{K}_{t}^{L}(\mathbf{c};\mathbf{i})\geq \mathcal{K}_{t}^{L}(\mathbf{c}%
;\mathbf{j})\geq \mathcal{K}_{t}^{L}(\mathbf{c})
\end{equation*}%
for every proper scoring rule $L$ and every $t\geq 1.$
\end{corollary}

The proof is given in Appendix \ref{sus-a:refined-refinement}.

\subsection{A Simple Way to Proper-Li Calibeat\label{sus:simple}}

We show that the simple calibeating procedure of Theorem 3 of Foster and
Hart (2023), whereby one forecasts the current state-average of the $b_{t}$%
-bin, is proper-Lipschitz calibeating. We then show that it is \emph{not}
proper calibeating, i.e., there are bounded proper scoring rules for which
it is not calibeating.

We write $\bar{a}_{t-1}^{\mathbf{b}}(b_{t})$ for the state average in the $%
b_{t}$-bin (of the $\mathbf{b}$-binning) over the periods up to $t-1$.

\begin{theorem}
\label{th:simple-calibeat-univ}Let $B$ be a finite set, and let $\zeta $ be
the deterministic $\mathbf{b}$-based forecasting procedure given by 
\begin{equation*}
c_{t}=\bar{a}_{t-1}^{\mathbf{b}}(b_{t})
\end{equation*}%
for every time $t\geq 1$ (if $t$ is the first time that $b_{t}$ is used,
take $c_{t}$ to be an arbitrary element of $C$). Then the procedure $\zeta $
is proper-Li $B$-calibeating; specifically,%
\begin{equation*}
0\leq \mathcal{B}_{t}^{L}(\mathbf{c})-\mathcal{R}_{t}^{L}(\mathbf{b})\leq
2|B|\frac{\ln t+1}{t},
\end{equation*}%
for all $t\geq 1$, all sequences $\mathbf{a}_{t}\in A^{t}$ and $\mathbf{b}%
_{t}\in B^{t}$, and all $1$-Lipschitz proper scoring rules $L$ (i.e., $L\in 
\mathcal{L}_{1}^{\mathrm{Li}}$).
\end{theorem}

As in Foster and Hart (2023), we define the \emph{online }$L$\emph{%
-refinement} score, for a pure binning sequence $\mathbf{i}$, as follows:%
\begin{equation*}
\widetilde{\mathcal{R}}_{t}^{L}\equiv \widetilde{\mathcal{R}}_{t}^{L}(%
\mathbf{i}):=\frac{1}{t}\sum_{s=1}^{t}D(a_{s},\bar{a}_{s-1}(i_{s}))
\end{equation*}%
(take $\bar{a}_{0}(i)$ to be an arbitrary point in $C$). In the refinement
score $\mathcal{R}_{t}^{L}$ one uses in each period $s$ an (offline) average
of the states, $\bar{a}_{t}(\cdot )$, taken over all periods from $1$ to $t$%
; this is replaced in $\widetilde{\mathcal{R}}_{t}^{L}$ by the corresponding
online average of the states, $\bar{a}_{s-1}(\cdot )$, taken only over the
past periods, from $1$ to $s-1$. We have:

\begin{proposition}
\label{p:online-R-S}Let $L$ be an $M$-Lipschitz proper scoring rule. Then%
\begin{equation*}
0\leq \widetilde{\mathcal{R}}_{t}^{L}(\mathbf{i})-\mathcal{R}_{t}^{L}(%
\mathbf{i})\leq 2M\frac{N_{t}}{t}\left( \ln \left( \frac{t}{N_{t}}\right)
+1\right) ,
\end{equation*}%
where $N_{t}:=\left\vert \{i_{s}:s\leq t\}\right\vert $ is the number of
bins used up to time $t.$
\end{proposition}

Before proving this proposition, we show that it directly establishes
Theorem \ref{th:simple-calibeat-univ}.

\medskip

\begin{proof}[Proof of Theorem \protect\ref{th:simple-calibeat-univ}]
Our choice of $c_{t}=\bar{a}_{t-1}^{\mathbf{b}}(b_{t})$ gives $\mathcal{B}%
_{t}^{L}(\mathbf{c})=\widetilde{\mathcal{R}}_{t}^{L}(\mathbf{b})$ for every $%
\mathbf{a},\mathbf{b}$, and $L$; the result then follows from Proposition %
\ref{p:online-R-S} and the fact that $N_{t}\leq \left\vert B\right\vert $.
\end{proof}

\subsubsection{Online vs. Offline Refinement\label{susus:online refinement}}

We now prove Proposition \ref{p:online-R-S}. This will follow from the
following:

\begin{proposition}
\label{p:online-offline}Let $x_{1},\ldots ,x_{n}\in C,$ let $L$ be a scoring
rule with divergence $D$, and define%
\begin{eqnarray*}
v_{n} &%
%TCIMACRO{\TeXButton{:=}{{\;:=\;}}}%
%BeginExpansion
{\;:=\;}%
%EndExpansion
&\frac{1}{n}\sum_{j=1}^{n}D(x_{j},\bar{x}_{n})\text{ and} \\
\widetilde{v}_{n} &%
%TCIMACRO{\TeXButton{:=}{{\;:=\;}}}%
%BeginExpansion
{\;:=\;}%
%EndExpansion
&\frac{1}{n}\sum_{j=1}^{n}D(x_{j},\bar{x}_{j-1}),
\end{eqnarray*}%
where $\bar{x}_{j}:=(x_{1}+...+x_{j})/j$ and $\bar{x}_{0}\in C$ is
arbitrary. Then%
\begin{equation}
\widetilde{v}_{n}-v_{n}=\frac{1}{n}\sum_{j=1}^{n}jD(\bar{x}_{j},\bar{x}%
_{j-1}).  \label{eq:v-online-error}
\end{equation}%
Moreover, if $L$ is an $M$-Lipschitz proper scoring rule, then%
\begin{equation*}
0\leq \widetilde{v}_{n}-v_{n}\leq 2M\frac{\ln n+1}{n}.
\end{equation*}
\end{proposition}

\begin{proof}
Let $\xi _{0}:=0$ and $\xi _{n}:=n(\widetilde{v}_{n}-v_{n})$ for $n\geq 1$;
canceling the $L(x_{j},x_{j})$ terms that appear in both sums yields 
\begin{equation*}
\xi _{n}=\sum_{j=1}^{n}L(x_{j},\bar{x}_{j-1})-\sum_{j=1}^{n}L(x_{j},\bar{x}%
_{n})=\sum_{j=1}^{n}L(x_{j},\bar{x}_{j-1})-nL(\bar{x}_{n},\bar{x}_{n})
\end{equation*}%
(for the second sum we have used $\bar{x}_{n}=(1/n)(x_{1}+\ldots +x_{n})$).
Put $\eta _{n}:=\xi _{n}-\xi _{n-1};$ we have%
\begin{equation*}
\eta _{n}=L(x_{n},\bar{x}_{n-1})-nL(\bar{x}_{n},\bar{x}_{n})+(n-1)L(\bar{x}%
_{n-1},\bar{x}_{n-1}).
\end{equation*}%
The sum of the first and third terms is $nL(\bar{x}_{n},\bar{x}_{n-1})$
(because $\bar{x}_{n}=(1/n)x_{n}+((n-1)/n)\bar{x}_{n-1}$), and so%
\begin{equation*}
\eta _{n}=n\left[ L(\bar{x}_{n},\bar{x}_{n-1})-L(\bar{x}_{n},\bar{x}_{n})%
\right] =nD(\bar{x}_{n},\bar{x}_{n-1}).
\end{equation*}%
Now $\xi _{n}=\sum_{j=1}^{n}\eta _{j},$ and so we have obtained the claimed
identity.

Properness gives $\eta _{j}\geq 0,$ and so $\xi _{n}\geq 0.$ The Lipschitz
condition gives, by (\ref{eq:L}),%
\begin{equation*}
\eta _{j}\leq jM\left\Vert \bar{x}_{j}-\bar{x}_{j-1}\right\Vert
^{2}=jM\left\Vert \frac{1}{j}(\bar{x}_{j-1}-x_{j})\right\Vert ^{2}\leq jM%
\frac{2}{j^{2}}=\frac{2M}{j}
\end{equation*}%
(we used $\left\Vert x-y\right\Vert ^{2}\leq 2$ for all $x,y\in C$).
Therefore%
\begin{equation*}
\xi _{n}=\sum_{j=1}^{n}\eta _{j}\leq 2M\sum_{j=1}^{n}\frac{1}{j}\leq 2M(\ln
n+1),
\end{equation*}%
completing the proof.
\end{proof}

\medskip

\begin{proof}[Proof of Proposition \protect\ref{p:online-R-S}]
For each bin $i$ with $n_{t}(i)>0,$ Proposition \ref{p:online-offline} gives%
\begin{equation*}
0\leq \frac{1}{n_{t}(i)}\sum_{s\leq t:i_{s}=i}D(a_{s},\bar{a}_{s-1}(i))-%
\frac{1}{n_{t}(i)}\sum_{s\leq t:i_{s}=i}D(a_{s},\bar{a}_{t}(i))\leq 2M\frac{%
\ln n_{t}(i)+1}{n_{t}(i)}.
\end{equation*}%
Averaging over all $i$ with weights $n_{t}(i)/t$ then yields%
\begin{equation*}
0\leq \widetilde{\mathcal{R}}_{t}^{L}-\mathcal{R}_{t}^{L}\leq 2M\frac{1}{t}%
\sum_{i}(\ln n_{t}(i)+1).
\end{equation*}%
By concavity of the logarithm, $\sum_{i}\ln n_{t}(i)\leq N_{t}\ln (t/N_{t})$
(because there are $N_{t}$ nonempty bins, and the sum of the $n_{t}(i)$ is $%
t $); this yields the claimed bound.
\end{proof}

\medskip

\noindent \textbf{Remarks.} \emph{(a)} In the quadratic case formula (\ref%
{eq:v-online-error}) yields Proposition 2 of Foster and Hart (2023) on the
online variance.

\emph{(b)} One may weaken the Lipschitz requirement to $\left\Vert \mathbf{L}%
(c)-\mathbf{L}(c^{\prime })\right\Vert \leq M\left\Vert c-c^{\prime
}\right\Vert ^{\alpha }$ for some $0<\alpha <1$ (this is usually called
\textquotedblleft $\alpha $-H\"{o}lder continuity\textquotedblright ). In
this case, we have $D^{L}(d,c)\leq M\left\Vert d-c\right\Vert ^{1+\alpha }$,
and the proof above then yields $\widetilde{v}_{n}-v_{n}=O(n^{-\alpha })$,
which implies that $\widetilde{\mathcal{R}}_{t}^{L}-\mathcal{R}%
_{t}^{L}=O(t^{-\alpha })$ and $\mathcal{B}_{t}^{L}(\mathbf{c})-\mathcal{R}%
_{t}^{L}(\mathbf{b})=O(t^{-\alpha })$.

\emph{(c)} Consider unbounded scoring rules. For the logarithmic scoring
rule, one may use the regularization of adding a positive constant to each
bin; see Appendix A.9 in Foster and Hart (2023f). However, this does \emph{%
not} work when the slope of $L$ towards the boundary of $C$ is much steeper.
For instance, let $L$ be the $\alpha $-power scoring rule with $\alpha =-1$,
a binary state space ($A=\{0,1\}$), and a single bin. If the sequence of
states is $x_{1}=(1,0)$ followed by $x_{j}=(0,1)$ for all $j\geq 2$, then%
\begin{equation*}
\eta _{j}=jD(\bar{x}_{j},\bar{x}_{j-1})=1+\frac{1}{(j-1)(j-2)^{2}}\geq 1
\end{equation*}%
for all $j\geq 3$, which shows that\footnote{%
Ignore the first $2$ periods where $\eta _{j}$ is infinite---these are the
\textquotedblleft regularization\textquotedblright\ periods in which we
\textquotedblleft seed\textquotedblright\ the bin---and start counting only
from $j=3$.} $\widetilde{v}_{n}-v_{n}\geq (n-2)/n\rightarrow 1$, and the
online refinement score is at a distance of about $1$ from the offline
refinement score.

\subsubsection{Proper-Li Calibeating vs. Proper Calibeating}

While the simple calibeating procedure is calibeating for every Lipschitz
proper scoring rule, the example below shows that it is not calibeating for
all bounded proper scoring rules.

\begin{example}
\label{ex:not-Lipschitz}In the one-dimensional case, where $A=\{0,1\}$,
identify elements $c=(c_{0},c_{1})\in C=\Delta (A)$ with the probability $%
c_{1}$ that $a=1$ (and then $c_{0}=1-c_{1}$). Consider a forecaster with a
constant sequence $\mathbf{b}$ (i.e., $b_{t}=b$ for all $t$). Let the state
sequence $\mathbf{a}$ be the alternating sequence $0,1,0,1,0,1,\ldots $.
Then $\bar{a}_{t}=1/2$ for even $t$ and $\bar{a}_{t}<1/2$ for odd $t$.
Consider the following bounded scoring rule:\footnote{%
It is generated by the utility function $u(a,x)=-\mathbf{1}_{a\neq x}$ for $%
a,x\in \{0,1\}$ (see Section \ref{s:utility} below), with $x^{\ast }(d)=1$
for $d\geq 1/2$ and $x^{\ast }(d)=0$ for $d<1/2$. The fact that the simple
calibeating procedure is not $L$-calibeating holds for every choice of
optimal decision at $d=1/2$ (where all decisions yield the same payoff);
indeed, if $x^{\ast }(1/2)<1$, then consider the state sequence $\mathbf{a}%
=(1,0,1,0,...)$.}%
\begin{equation*}
L(d,c):=\left\{ 
%TCIMACRO{\TeXButton{\arraystretch=0.5}{\renewcommand{\arraystretch}{0.5}}}%
%BeginExpansion
\renewcommand{\arraystretch}{0.5}%
%EndExpansion
\begin{tabular}{lll}
$1-d,$ &  & if $c\geq 1/2$ \\ 
&  &  \\ 
$d,$ &  & if $c<1/2$.%
\end{tabular}%
%TCIMACRO{\TeXButton{\arraystretch=1}{\renewcommand{\arraystretch}{1}}}%
%BeginExpansion
\renewcommand{\arraystretch}{1}%
%EndExpansion
\right. .
\end{equation*}%
$L$ is bounded but is not continuous, and thus not Lipschitz, since $%
L(0,1/2)=1$ and $L(0,1/2-\varepsilon )=0$ for $\varepsilon >0$. The $L$%
-divergence is%
\begin{equation*}
D^{L}(d,c)=\left\{ 
%TCIMACRO{\TeXButton{\arraystretch=0.5}{\renewcommand{\arraystretch}{0.5}}}%
%BeginExpansion
\renewcommand{\arraystretch}{0.5}%
%EndExpansion
\begin{tabular}{lll}
$0,$ &  & if ($d\geq 1/2$ and $c\geq 1/2$) or ($d<1/2$ and $c<1/2$), \\ 
&  &  \\ 
$|2d-1|,$ &  & if ($d\geq 1/2$ and $c<1/2$) or ($d<1/2$ and $c\geq 1/2$).%
\end{tabular}%
%TCIMACRO{\TeXButton{\arraystretch=1}{\renewcommand{\arraystretch}{1}}}%
%BeginExpansion
\renewcommand{\arraystretch}{1}%
%EndExpansion
\right. ,
\end{equation*}%
and the $L$-entropy is%
\begin{equation*}
H^{L}(d)=\min \{d,1-d\}.
\end{equation*}

For the sequence $c_{t}=\bar{a}_{t-1}$ (taking, say, $\bar{a}_{0}=1/2$), we
then have $\mathcal{B}_{t}^{L}(\mathbf{c})=1$ and $\mathcal{R}_{t}^{L}(%
\mathbf{b})\leq 1/2$ for all $t$. Indeed, $H^{L}(a_{s})=0$ for every $s$,
and so $\mathcal{H}_{t}^{L}=(1/t)\sum_{s\leq t}H^{L}(a_{s})=0$; then $%
\mathcal{R}_{t}^{L}(\mathbf{b})=H^{L}(\bar{a}_{t})-\mathcal{H}_{t}^{L}=H^{L}(%
\bar{a}_{t})\leq 1/2$ (in fact, $\bar{a}_{t}\rightarrow 1/2$ implies $%
\mathcal{R}_{t}^{L}(\mathbf{b})\rightarrow 1/2$); finally, $%
D^{L}(a_{s},c_{s})=D^{L}(a_{s},\bar{a}_{s-1})=1$ for every $s$ (because we
have $a_{s}=0$ and $\bar{a}_{s-1}=1/2$ for odd $s$, while $a_{s}=1$ and $%
\bar{a}_{s-1}<1/2$ for even $s$, so in both cases $D^{L}(a_{s},\bar{a}%
_{s-1})=1$), and so $\mathcal{B}_{t}^{L}(\mathbf{c})=(1/t)\sum_{s\leq
t}D^{L}(a_{s},\bar{a}_{s-1})=1$.

The simple calibeating procedure, which is proper-Li calibeating, is thus 
\emph{not} $L$-calibeating for the above bounded proper scoring rule $L$,
and hence \emph{not} proper calibeating.
\end{example}

\subsection{Proper-Li Calibeating by a Proper Continuously-Calibrated
Deterministic Procedure\label{sus:continuous}}

The result of this section is the proper counterpart of Theorems 6 and 12 of
Foster and Hart (2023, 2023ff). Since general binnings need not refine the
standard by-forecast binning (because each bin may well contain forecasts
with different values), and the decomposition $\mathcal{B}^{L}=\mathcal{K}%
^{L}+\mathcal{R}^{L}$ is no longer valid, we use the approximate
Decomposition Theorem \ref{th:decompose-delta} instead of the exact
Decomposition Theorem \ref{th:S-decompose}. Recall Section \ref%
{susus:fractional}: a continuous binning $\Pi =(w_{i})_{i\in I}$ is $\delta $%
\emph{-local} if for every $i$ there is a $y^{i}\in C$ such that $\{c\in
C:w_{i}(c)>0\}\subseteq \bar{B}(y^{i};\delta )$; i.e., all forecasts in bin $%
i$ are within $\delta $ of $y^{i}$. For every forecasting sequence $\mathbf{c%
}$, the resulting general binning sequence $\Pi (\mathbf{c})$ is then $%
\delta $-local with respect to $\mathbf{c}$. Finally, the continuous binning 
$\Pi ^{\ast }$ is given by Proposition 12 of Foster and Hart (2023ff).

\begin{theorem}
\label{th:calibeat by cont cal S}Let $B$ be a finite set. Then there exists
a \emph{deterministic} $\mathbf{b}$-based forecasting procedure $\zeta $
that is proper-Li $B$-calibeating and proper continuously-calibrated.
Specifically: first, for every continuous binning $\Pi $ there is a
deterministic $\mathbf{b}$-based forecasting procedure $\zeta $ such that%
\footnote{%
The notation $o(1)$ denotes a function $\eta (t)$ that satisfies $\eta
(t)\rightarrow 0$ as $t\rightarrow \infty $.}%
\begin{equation}
\mathcal{K}_{t}^{L}(\mathbf{c};\mathbf{b}\times \Pi (\mathbf{c}))\leq o(1);
\label{eq:34}
\end{equation}%
and second, for the continuous binning $\Pi ^{\ast }$ of Foster and Hart
(2023ff), condition (\ref{eq:34}) implies that 
\begin{equation*}
\mathcal{B}_{t}^{L}(\mathbf{c})\leq \mathcal{R}_{t}^{L}(\mathbf{b})+o(1),
\end{equation*}%
and that $\zeta $ is continuously $L$-calibrated. All these hold as $%
t\rightarrow \infty $ uniformly over all sequences $\mathbf{a}$ and $\mathbf{%
b}$ and $1$-Lipschitz proper scoring rules $L$ (i.e., $L\in \mathcal{L}_{1}^{%
\mathrm{Li}}$).
\end{theorem}

\begin{proof}
For every continuous binning $\Pi $ the procedure of Theorem 12 in Foster
and Hart (2023ff) yields (\ref{eq:34}) for the quadratic scoring rule, and
thus uniformly for all $1$-bounded proper scoring rules\footnote{%
The Lipschitz restriction will be needed only for the next argument.} $L$ in 
$\mathcal{L}_{1}$ by Proposition \ref{p:K-KS}.

Now consider $\Pi ^{\ast }$, which contains the binning $\Pi _{0}$ of Foster
and Hart (2021) and a sequence $(\Pi _{n})_{n\geq 1}$ of $\delta _{n}$-local
continuous binnings with $\delta _{n}>0$ converging to $0.$ As shown in the
proof of Theorem 12 in Foster and Hart (2023ff), condition (\ref{eq:34}) for 
$\Pi ^{\ast }$ implies that for each $n\geq 0$ we have%
\begin{equation*}
\mathcal{K}_{t}(\mathbf{c};\mathbf{b}\times \Pi _{n}(\mathbf{c}))\leq o(1)
\end{equation*}%
(see (41) there). For $n=0$ this implies that $\mathcal{K}_{t}(\mathbf{c}%
;\Pi _{0}(\mathbf{c}))\leq \mathcal{K}_{t}(\mathbf{c};\mathbf{b}\times \Pi
_{0}(\mathbf{c}))\leq o(1)$ (by the convexity of the squared norm, since $%
\mathbf{b}\times \Pi _{0}(\mathbf{c})$ refines\footnote{\label{ftn:K(f) ge
K(g)}While for the standard quadratic calibration the inequality $K(\mathbf{c%
};\mathbf{f})\geq K(\mathbf{c};\mathbf{g})$ holds for general (fractional)
binnings $\mathbf{f}$ and $\mathbf{g}$ with $\mathbf{f}$ a refinement of $%
\mathbf{g}$ (by the convexity of $D(a,c)=\left\Vert a-c\right\Vert ^{2}$ in $%
a-c$), it need not hold for other bounded proper scoring rules. For binnings
that refine $\mathbf{c}$, see Corollary \ref{c:refine-S}. See also Foster
and Hart (2023ff, Appendix A.5 and Errata).} $\Pi _{0}(\mathbf{c})$), and so 
$\zeta $ is continuously calibrated (by Proposition 3 in Foster and Hart
2021), and thus proper continuously-calibrated (by Theorem \ref%
{th:proper-calibration}). For $n\geq 1$, using the approximate decomposition
of Theorem \ref{th:decompose-delta} for $1$-Lipschitz proper scoring rules $%
L $ (indeed, the general binning sequence $\Pi _{n}(\mathbf{c})$ is $\delta
_{n}$-local with respect to $\mathbf{c}$, and thus so is its refinement $%
\mathbf{b}\times \Pi _{n}(\mathbf{c})$) and, again, Proposition \ref{p:K-KS}%
, we get%
\begin{eqnarray*}
\mathcal{B}_{t}^{L}(\mathbf{c})-\mathcal{R}_{t}^{L}(\mathbf{b}\times \Pi
_{n}(\mathbf{c})) &\leq &\mathcal{K}_{t}^{L}(\mathbf{c};\mathbf{b}\times \Pi
_{n}(\mathbf{c}))+2\delta _{n} \\
&\leq &\mathcal{K}_{t}(\mathbf{c};\mathbf{b}\times \Pi _{n}(\mathbf{c}%
))+2\delta _{n}\leq 2\delta _{n}+o(1),
\end{eqnarray*}%
uniformly for all $L$ in $\mathcal{L}_{1}^{\mathrm{Li}}$. Since $\mathcal{R}%
_{t}^{L}(\mathbf{b}\times \Pi _{n}(\mathbf{c}))\leq \mathcal{R}_{t}^{L}(%
\mathbf{b})$ by Proposition \ref{p:refine-S}, we get%
\begin{equation*}
\mathcal{B}_{t}^{L}(\mathbf{c})-\mathcal{R}_{t}^{L}(\mathbf{b})\leq 2\delta
_{n}+o(1).
\end{equation*}%
Therefore, $\mathcal{B}_{t}^{L}(\mathbf{c})-\mathcal{R}_{t}^{L}(\mathbf{b}%
)\leq 3\delta _{n}$ for all $t$ large enough; since $\delta _{n}\rightarrow
0,$ this yields $\mathcal{B}_{t}^{L}(\mathbf{c})-\mathcal{R}_{t}^{L}(\mathbf{%
b})\leq o(1).$
\end{proof}

\subsection{Proper Multicalibeating\label{sus:multi}}

Suppose that there are $N\geq 1$ forecasting sequences, $\mathbf{b}%
^{n}=(b_{t}^{n})_{t\geq 1}$ for $n=1,2,\ldots ,N$. We assume that each $%
\mathbf{b}^{n}$ uses only finitely many forecasts: there is a finite set $%
B^{n}$ such that $b_{t}^{n}\in B^{n}$ for all $t\geq 1$. Set $\mathbf{b}=(%
\mathbf{b}^{1},\ldots ,\mathbf{b}^{N});$ we seek a $\mathbf{b}$-based
forecasting procedure---i.e., $c_{t}$ is determined after all the $%
b_{t}^{1},\ldots ,b_{t}^{N}$ are announced (and hence is a function of $%
\mathbf{a}_{t-1},\mathbf{c}_{t-1},\mathbf{b}_{t}^{1},\ldots ,\mathbf{b}%
_{t}^{N}$)---that simultaneously calibeats all the $\mathbf{b}^{n}$
sequences under the relevant classes of proper scoring rules. By applying
the results of the previous section to the joint binning $\mathbf{b}%
^{1}\times \cdots \times \mathbf{b}^{N}$ we obtain the analogue of Theorem 7
in Foster and Hart (2023):

\begin{theorem}
\label{th:multi}Let $B^{1},\ldots ,B^{N}$ be finite sets. Then:

\begin{description}
\item[(i)] For every finite $\delta $-grid $C_{\delta }$ of $C$ there exists
a stochastic $(\mathbf{b}^{1},\ldots ,\mathbf{b}^{N})$-based $C_{\delta }$%
-forecasting procedure $\zeta $ that is proper complete $(\sqrt{\delta }%
,B^{n})$-calibeating, and thus proper $(\sqrt{\delta },B^{n})$-calibeating,
for all $n=1,\ldots ,N$, and is proper $\sqrt{\delta }$-calibrated.
Moreover, $\zeta $ may be taken to be $\delta $-almost deterministic.

\item[(ii)] There exists a simple deterministic $(\mathbf{b}^{1},\ldots ,%
\mathbf{b}^{N})$-based forecasting procedure $\zeta $ that is proper-Li $%
B^{n}$-calibeating for all $n=1,\ldots ,N;$ specifically, the forecast of $%
\zeta $ in period $t$ is $c_{t}=\bar{a}_{t-1}^{\mathbf{b}^{1},\ldots ,%
\mathbf{b}^{N}}(b_{t}^{1},\ldots ,b_{t}^{N})$, the average of the states in
all past periods $s\leq t-1$ where the combination $(b_{t}^{1},\ldots
,b_{t}^{N})$ was used (if $t$ is the first period in which $%
(b_{t}^{1},\ldots ,b_{t}^{N})$ is used, take $c_{t}\in C$ to be arbitrary).

\item[(iii)] There exists a deterministic $(\mathbf{b}^{1},\ldots ,\mathbf{b}%
^{N})$-based $C$-forecasting procedure $\zeta $ that is proper-Li $B^{n}$%
-calibeating for all $n=1,\ldots ,N$, and is proper continuously-calibrated.
\end{description}
\end{theorem}

\section{Decision-Making Under Uncertainty\label{s:utility}}

Consider a decision maker with a utility function $u:A\times X\rightarrow 
\mathbb{R}$, where $A$ is a finite set of states (\textquotedblleft states
of nature\textquotedblright ) and $X$ is a set of \textquotedblleft
decisions.\textquotedblright\ For every probability distribution $d$ on $A$,
i.e., $d\in C=\Delta (A)$, let\footnote{%
We slightly abuse notation and continue to use $u$ also for its linear
extension (thus, $u(a,x)$ and $u(\mathbf{1}_{a},x)$ are the same).} 
\begin{equation*}
u(d,x):=\mathbb{E}_{a\sim d}\left[ u(a,x)\right]
\end{equation*}%
be the expected utility for the decision $x$ in $X$, and let 
\begin{equation*}
\overline{u}(d):=\sup_{x\in X}u(d,x)
\end{equation*}%
be the highest expected utility.

We assume that the maximum is always attained (which is the case, for
instance, when $X$ is compact and $u$ is continuous). Let $x^{\ast }(d)\in X$
be an \emph{optimal (maximizing) decision} for the distribution of states $%
d\in C$; i.e., $\overline{u}(d)=u(d,x^{\ast }(d))$.

To evaluate a forecast $c$ in $C$, let the loss from using $c$ be the
disutility incurred by choosing the optimal decision $x^{\ast }(c)$ when the
realized state is $a$:%
\begin{equation*}
L_{A}^{u}(a,c):=-u(a,x^{\ast }(c)).
\end{equation*}%
The expected loss when the true probability distribution is $d\in C$ is then%
\begin{equation*}
L^{u}(d,c):=\mathbb{E}_{a\sim d}\left[ L_{A}^{u}(a,c)\right] =-\mathbb{E}%
_{a\sim d}\left[ u(a,x^{\ast }(c))\right] =-u(d,x^{\ast }(c)).
\end{equation*}%
We refer to $L^{u}$ as the scoring rule \emph{induced} by $u$ (more
precisely, by $u$ and $x^{\ast }$; see Remark (c) below). The equivalence
between utility maximizing and proper scoring is well known (see, e.g.,
Savage 1971).

\begin{proposition}
\label{p:u-L}$L$ is a proper scoring rule if and only if there is a utility
function $u$ (with optimal decisions $x^{\ast }$) such that the induced
scoring rule is $L$, i.e., $L=L^{u}$.
\end{proposition}

\begin{proof}
Given $u$ and $x^{\ast }$, the optimality of $x^{\ast }(d)$ yields $%
u(d,x^{\ast }(d))\geq u(d,x^{\ast }(c))$, i.e., $L^{u}(d,d)\leq L^{u}(d,c)$.

Conversely, given a proper scoring rule $L$ set $X:=C$ and $%
u(a,c):=-L_{A}(a,c)$, then $x^{\ast }(d)=d$ is optimal by properness, and it
yields $L^{u}=L$.
\end{proof}

\medskip

\noindent \textbf{Remarks.} \emph{(a)} $L^{u}$ is bounded when $u$ is
bounded (it suffices that $u$ is bounded on the range of $x^{\ast }$).

\emph{(b)} Stronger regularity conditions are needed in order for the
scoring rule $L^{u}$ to be Lipschitz. For instance, let $X$ be a compact
convex subset of a Euclidean space, and for every $a$ let the function $%
u(a,\cdot )$ be continuously differentiable and strongly concave\footnote{%
A function $f(c)$ is \emph{strongly concave} if its curvature is bounded
away from zero; i.e., there exists $\delta >0$ such that $f(c)+\delta
\left\Vert c\right\Vert ^{2}$ is concave.} on $X$; then the optimal decision 
$x^{\ast }(d)$ is unique and Lipschitz in $d$, and so $L^{u}$ is Lipschitz.

\emph{(c)} The choice of the optimal decision $x^{\ast }(d)$ when multiple
maximizers exist may affect specific values of the scoring rule $L^{u}$, but
does not alter its fundamental properties. All subsequent statements hold
for any selection of optimal decisions $x^{\ast }$.

\subsection{Regret, Calibration, and Calibeating\label{sus:u-regret}}

We now study the relation between proper calibration and proper calibeating,
and no-regret when best replying to forecasts.

\subsubsection{Regret and Decision Value\label{susus:u-decision value}}

Let $u$ be a utility function with optimal decision mapping $x^{\ast }$ and
induced proper scoring rule $L^{u}$. A forecasting sequence $\mathbf{c}%
=(c_{t})_{t\geq 1}$ generates a sequence of decisions $\mathbf{x}%
=(x_{t})_{t\geq 1}$ by best replying to the forecast, i.e., $x_{t}=x^{\ast
}(c_{t})$ for every $t$. The resulting \emph{average utility} up to time $t$
is%
\begin{equation*}
U_{t}(\mathbf{c}):=\frac{1}{t}\sum_{s=1}^{t}u(a_{s},x^{\ast }(c_{s})).
\end{equation*}%
The \emph{regret for }$u$ \emph{of best replying to the sequence} $\mathbf{c}
$ (the\emph{\ }$u$\emph{-regret of }$c$\ for short) is the maximal possible
increase in average utility obtainable by using a map $\xi :C\rightarrow X$
from forecasts to decisions instead of $x^{\ast }$:%
\begin{equation*}
\text{\textsc{Reg}}_{t}^{u}(\mathbf{c}):=V_{t}^{u}(\mathbf{c})-U_{t}(\mathbf{%
c}),
\end{equation*}%
where%
\begin{equation*}
V_{t}^{u}(\mathbf{c}):=\max_{\xi :C\rightarrow X}\frac{1}{t}%
\sum_{s=1}^{t}u(a_{s},\xi (c_{s}))
\end{equation*}%
is the \emph{decision} \emph{value} embodied in the sequence $\mathbf{c},$
more precisely, in the $\mathbf{c}$-binning.\footnote{%
We will see in Proposition \ref{p:reg=K} below that the maximum is attained.}
This is the maximal average utility that can be obtained using the
partitioning given by the $\mathbf{c}$ sequence; that is, the decision $%
x_{s} $ at time $s$ is determined by the forecast $c_{s}$ only (this maximal
value can only be computed ex post, i.e., offline). Thus, $c_{s}$ need not
be taken at face value as a forecast, but rather as a \textquotedblleft
signal\textquotedblright\ on which one can base one's decision.

\medskip

\noindent \textbf{Remark. }The notion of regret here is a stronger version
of the so-called \textquotedblleft swap regret.\textquotedblright\ While
swap regret considers all mappings $\phi :X\rightarrow X$ in which the
decision $x$ is replaced throughout by $\phi (x)$, our notion allows for
mappings $\xi :C\rightarrow X$ in which decisions depend directly on
forecasts. Thus, if two distinct forecasts $c\neq c^{\prime }$ yield the
same optimal decision, i.e., $x^{\ast }(c)=x^{\ast }(c^{\prime })=x$, swap
regret would force the same alternative decision $\phi (x)$ for both $c$ and 
$c^{\prime }$, whereas our notion allows for distinct alternative decisions $%
\xi (c)\neq \xi (c^{\prime }).$ For a simple example, in the one-dimensional
case suppose that the best replies are $x_{0}$ if $c\leq 1/2$ and $x_{1}$ if 
$c\geq 1/2$ (with, say, $x_{0}$ selected at $c=1/2$). For swap regret, the
precise value of the forecast is then irrelevant; only which side of $1/2$
it lies on matters.\footnote{%
Thus, one can have the state average below $1/2$ in the periods when $%
c_{s}\leq 1/2$ and the state average above $1/2$ in the periods when $%
c_{s}>1/2$, and hence zero swap regret, while, say, the state average in the
periods when $c_{s}=0.7$ is $0.2$, which makes our regret strictly positive.}
The regret here uses the full distinction among periods provided by the
forecasting sequence $c_{t}$, rather than the possibly coarser distinction
provided by the induced decision sequence $x_{t}=x^{\ast }(c_{t})$ used in
the standard swap regret notion.

\subsubsection{Proper Calibration and Universal No Regret\label{susus:u-no
regret}}

It turns out that the $u$-regret is precisely the $L^{u}$-calibration score.
We show this in a more general setup, which will be used in the next section
as well.

Let $\mathbf{i}$ be a binning sequence (with $i_{t}\in I$) that is finer
than the binning-by-forecast sequence $\mathbf{c}$ (i.e., all entries in an $%
i$-bin have identical $c$). When considering alternative decision sequences
we are now allowed to use not just the forecasts $c_{s}$, but also the
additional information embodied in the binning $i_{s}$. The $u$\emph{-regret
of }$\mathbf{c}$\emph{\ with respect to }$\mathbf{i}$ is thus%
\begin{equation*}
\text{\textsc{Reg}}_{t}^{u}(\mathbf{c};\mathbf{i}):=V_{t}^{u}(\mathbf{i}%
)-U_{t}(\mathbf{c}),
\end{equation*}%
where%
\begin{equation*}
V_{t}^{u}(\mathbf{i}):=\max_{\xi :I\rightarrow X}\frac{1}{t}%
\sum_{s=1}^{t}u(a_{s},\xi (i_{s}))
\end{equation*}%
is the \emph{decision value} of the sequence $\mathbf{i}$. The regret is
always nonnegative,%
\begin{equation*}
\text{\textsc{Reg}}_{t}^{u}(\mathbf{c};\mathbf{i})\geq 0,
\end{equation*}%
because $\mathbf{i}$ refines $\mathbf{c}$ and so the set of mappings
includes best replying to the forecast; moreover, the finer the binning the
higher the regret (because the maximization is taken over a larger set of
mappings). When $\mathbf{i}=\mathbf{c}$, the definition reduces to the
regret of the previous section: \textsc{Reg}$_{t}^{u}(\mathbf{c};\mathbf{c}%
)= $\textsc{Reg}$_{t}^{u}(\mathbf{c})$.

\begin{proposition}
\label{p:reg=K}Let $\mathbf{c}$ be a forecasting sequence, $\mathbf{i}$ a
binning sequence, and $u$ a utility function with induced proper scoring
rule $L^{u}$. Then%
\begin{eqnarray*}
U_{t}(\mathbf{c}) &=&-\mathcal{B}_{t}^{L^{u}}(\mathbf{c})-\mathcal{H}%
_{t}^{L^{u}}\text{\ \ and} \\
V_{t}^{u}(\mathbf{i}) &=&\frac{1}{t}\sum_{s=1}^{t}u(a_{s},x^{\ast }(\bar{a}%
_{t}(i_{s})))=-\mathcal{R}_{t}^{L^{u}}(\mathbf{i})-\mathcal{H}_{t}^{L^{u}},
\end{eqnarray*}%
and so if $\mathbf{i}$ is a refinement of $\mathbf{c}$ then%
\begin{equation*}
\text{\textsc{Reg}}_{t}^{u}(\mathbf{c};\mathbf{i})=\mathcal{K}_{t}^{L^{u}}(%
\mathbf{c};\mathbf{i});
\end{equation*}%
in particular, for the by-forecast binning $\mathbf{i}=\mathbf{c}$, we have%
\begin{equation*}
\text{\textsc{Reg}}_{t}^{u}(\mathbf{c})=\mathcal{K}_{t}^{L^{u}}(\mathbf{c}).
\end{equation*}
\end{proposition}

\begin{proof}
By the definition of $L\equiv L^{u}$ we have $u(a_{s},x^{\ast
}(c_{s}))=-L(a_{s},c_{s})$ for every $s$, and so, by (\ref{eq:B=avg-L - H}),%
\begin{equation*}
U_{t}(\mathbf{c})=-\frac{1}{t}\sum_{s=1}^{t}L(a_{s},c_{s})=-\mathcal{B}%
_{t}^{L}(\mathbf{c})-\mathcal{H}_{t}^{L}.
\end{equation*}%
Next, 
\begin{eqnarray*}
\frac{1}{t}\sum_{s=1}^{t}u(a_{s},\xi (i_{s})) &=&\sum_{i}\left( \frac{%
n_{t}(i)}{t}\right) \frac{1}{n_{t}(i)}\sum_{s\leq t:i_{s}=i}u(a_{s},\xi (i))
\\
&=&\sum_{i}\left( \frac{n_{t}(i)}{t}\right) u(\bar{a}_{t}(i),\xi (i)).
\end{eqnarray*}%
This is maximized when $\xi (i)=x^{\ast }(\bar{a}_{t}(i))$ for each $i$, and
so%
\begin{eqnarray*}
V_{t}^{u}(\mathbf{i}) &=&\frac{1}{t}\sum_{s=1}^{t}u(a_{s},x^{\ast }(\bar{a}%
_{t}(i_{s})))=\sum_{i}\left( \frac{n_{t}(i)}{t}\right) u(\bar{a}%
_{t}(i),x^{\ast }(\bar{a}_{t}(i))) \\
&=&-\sum_{i}\left( \frac{n_{t}(i)}{t}\right) L(\bar{a}_{t}(i),\bar{a}%
_{t}(i))=-\mathcal{R}_{t}^{L}(\mathbf{i})-\mathcal{H}_{t}^{L}
\end{eqnarray*}%
(see (\ref{eq:R-H-H})). The Decomposition Theorem \ref{th:S-decompose} then
yields the final equality.
\end{proof}

\medskip

Thus, having no regret is the same as being calibrated. When considering all
bounded utility functions---equivalently, all bounded proper scoring rules
(by Proposition \ref{p:u-L} and Remark (a) there)---this yields:

\begin{quote}
A forecasting procedure is proper calibrated if and only if every decision
maker with bounded utility has no regret when best replying to the forecasts.%
\footnote{%
The no-regret statement should be understood to hold uniformly over
utilities. This decision-theoretic characterization is closely related to
Foster and Vohra (1997, 1998); see also Hu and Wu (2024), who analyze the
corresponding worst-case decision loss in the one-dimensional (binary state)
case.}
\end{quote}

\noindent Using \textquotedblleft universal\textquotedblright\ to refer to 
\emph{all} bounded utility functions, we will say that the procedure $\sigma 
$ has \emph{universal regret} of at most $\varepsilon $ if the regrets are
guaranteed to be $\leq \varepsilon ^{2}$, uniformly over states of nature
and utilities; formally,%
\begin{equation}
\varlimsup_{t\rightarrow \infty }\left( \sup_{u\in \mathcal{U}_{1}}\sup_{%
\mathbf{a}_{t}\in A^{t}}\mathbb{E}\left[ \text{\textsc{Reg}}_{t}^{u}(\mathbf{%
c})\right] \right) \leq \varepsilon ^{2},  \label{eq:univ-regret}
\end{equation}%
where $\mathcal{U}_{1}$ denotes all utility functions $u$ with induced
proper scoring rule $L^{u}$ that is $1$-bounded (i.e., $L^{u}\in \mathcal{L}%
_{1}$). For $\varepsilon =0$ we thus get:

\begin{theorem}
\label{th:proper calib univ no regret}A forecasting procedure $\sigma $ is
proper calibrated if and only if it has universal no regret.
\end{theorem}

Since proper calibration is equivalent to calibration, we have:%
\begin{equation*}
\begin{tabular}{|c|}
\hline
$\text{calibration\ }\Longleftrightarrow \;\text{proper calibration\ }%
\Longleftrightarrow \;\text{universal no regret}$ \\ \hline
\end{tabular}%
\end{equation*}

Moreover:

\begin{corollary}
Let $\sigma $ be a proper $\varepsilon $-calibrated procedure. Then%
\begin{equation*}
\mathbb{E}\left[ \text{\textsc{Reg}}_{t}^{u}(\mathbf{c})\right] \leq
\varepsilon ^{2}+o(1)
\end{equation*}%
as $t\rightarrow \infty $, uniformly over all state sequences $\mathbf{a}$
and all utility functions $u$ with induced proper scoring rule $L^{u}$ that
is $1$-bounded.
\end{corollary}

\medskip

This applies, for instance, to the result of Theorem \ref%
{th:proper-calibration}.

\subsubsection{Proper Complete Calibeating and Universal Subsuming\label%
{susus:u-subsuming}}

We next consider the setup with a reference\ sequence $\mathbf{b}$ with $%
b_{t}$ in an arbitrary fixed finite set $B$ (only its induced binning will
matter for now; see also Appendix \ref{sus-a:improvement}).

Given a utility function $u,$ best replying to the forecasts $\mathbf{c}$ is
now measured not just against the decision value $V(\mathbf{c})$ of $\mathbf{%
c}$, but rather against the possibly higher target that is now available:
the \emph{joint} \emph{decision value of }$\mathbf{b}$\emph{\ and} $\mathbf{c%
}$, namely 
\begin{equation*}
V_{t}^{u}(\mathbf{b}\times \mathbf{c})=\max_{\xi :B\times C\rightarrow X}%
\frac{1}{t}\sum_{s=1}^{t}u(a_{s},\xi (b_{s},c_{s})).
\end{equation*}%
This is the decision value of the joint binning $\mathbf{b}\times \mathbf{c}$%
; we will call it the \emph{\ }$(\mathbf{b},\mathbf{c})$\emph{-joint
decision value}. This is based on the joint binning $\mathbf{b}\times 
\mathbf{c}$ induced by $\mathbf{b}$ and $\mathbf{c}$ together: the decision $%
x_{s}$ at time $s$ may depend on both $b_{s}$ and $c_{s}$ (and nothing
else). The resulting $u$\emph{-regret of }$\mathbf{c}$ \emph{with respect to}
$\mathbf{b}$ \emph{and} $\mathbf{c}$ (the $(\mathbf{b},\mathbf{c})$\emph{%
-joint} \emph{regret of} $\mathbf{c}$\ for short) is then%
\begin{equation*}
\text{\textsc{Reg}}_{t}^{u}(\mathbf{c};\mathbf{b}\times \mathbf{c}%
)=V_{t}^{u}(\mathbf{b}\times \mathbf{c})-U_{t}(\mathbf{c}).
\end{equation*}%
A forecasting procedure with forecasts $\mathbf{c}$ has \emph{universal }$(%
\mathbf{b},\mathbf{c})$\emph{-joint} \emph{regret} of at most $\varepsilon $
if 
\begin{equation*}
\varlimsup_{t\rightarrow \infty }\left( \sup_{u\in \mathcal{U}_{1}}\sup_{%
\mathbf{a}_{t}\in A^{t},\mathbf{b}_{t}\in B^{t}}\mathbb{E}\left[ \text{%
\textsc{Reg}}_{t}^{u}(\mathbf{c};\mathbf{b}\times \mathbf{c})\right] \right)
\leq \varepsilon ^{2}.
\end{equation*}%
Clearly $V_{t}^{u}(\mathbf{b}\times \mathbf{c})\geq V_{t}^{u}(\mathbf{c})$
(the joint decision value can only be higher, since the mappings $\xi $ may
also use $b_{s}$), and so the joint regret is weakly higher than the
standard regret: \textsc{Reg}$_{t}^{u}(\mathbf{c};\mathbf{b}\times \mathbf{c}%
)\geq $\textsc{Reg}$_{t}^{u}(\mathbf{c})$. Therefore (proper) calibration,
as in Theorem \ref{th:proper calib univ no regret}, no longer suffices. What
is needed is proper complete calibeating (i.e., proper calibeating of the
joint):

\begin{theorem}
A $\mathbf{b}$-based forecasting procedure $\zeta $ is proper complete $B$%
-calibeating if and only if it has universal no $(\mathbf{b},\mathbf{c})$%
-joint\emph{\ }regret.
\end{theorem}

\begin{proof}
Proper complete calibeating is the same as proper calibration on the joint
(because $\mathcal{B}_{t}^{L}(\mathbf{c})-\mathcal{R}_{t}^{L}(\mathbf{b}%
\times \mathbf{c})=\mathcal{K}_{t}^{L}(\mathbf{c};\mathbf{b}\times \mathbf{c}%
)$; cf. Remark (a) in Section \ref{sus:joint}); and this is the same as no
regret on the joint (\textsc{Reg}$_{t}^{u}(\mathbf{c};\mathbf{b}\times 
\mathbf{c})=\mathcal{K}_{t}^{L^{u}}(\mathbf{c};\mathbf{b}\times \mathbf{c})$
by Proposition \ref{p:reg=K}).
\end{proof}

Since proper complete calibeating is equivalent to complete calibeating, we
have:%
\begin{equation*}
\begin{tabular}{|ccl|}
\hline
complete $\text{calibeating}$ & $\Longleftrightarrow $ & $\text{proper
complete calibeating\ }$ \\ 
& $\Longleftrightarrow $ & $\text{universal no joint regret}$ \\ \hline
\end{tabular}%
\end{equation*}

Corresponding approximate versions hold (with the appropriate $\varepsilon $
to $\sqrt{\varepsilon }$ adaptations, as in Theorem \ref%
{th:proper-calibration}).

Next, let us characterize the no-joint-regret condition. We define the \emph{%
decision value added by} $\mathbf{b}$ \emph{to} $\mathbf{c}$ by%
\begin{eqnarray*}
\Delta V_{t}^{u}(\mathbf{b\,}|\,\mathbf{c}) &%
%TCIMACRO{\TeXButton{:=}{{\;:=\;}}}%
%BeginExpansion
{\;:=\;}%
%EndExpansion
&V_{t}^{u}(\mathbf{b}\times\mathbf{c})-V_{t}^{u}(\mathbf{c}) \\
&=&\max_{\xi :B\times C\rightarrow X}\frac{1}{t}\sum_{s=1}^{t}u(a_{s},\xi
(b_{s},c_{s}))-\max_{\xi :C\rightarrow X}\frac{1}{t}\sum_{s=1}^{t}u(a_{s},%
\xi (c_{s}))\geq 0.
\end{eqnarray*}%
This measures the increase in decision value from using the signals in $%
\mathbf{b}$ in addition to those in $\mathbf{c}$. When the value added
vanishes asymptotically, i.e., $\Delta V_{t}^{u}(\mathbf{b\,}|\,\mathbf{c}%
)\rightarrow 0$ as $t\rightarrow \infty $, we will say that $\mathbf{c}$ 
\emph{subsumes} $\mathbf{b}$ (for $u$); in this case, using only $\mathbf{c}$
and ignoring $\mathbf{b}$ does not lose decision value (in the limit).
Requiring this for all utility functions, we say that a procedure with
forecasting sequence $\mathbf{c}$ is \emph{universally }$B$\emph{-subsuming}
if it is guaranteed to subsume every $B$-valued reference sequence; formally,%
\begin{equation*}
\lim_{t\rightarrow \infty }\left( \sup_{u\in \mathcal{U}_{1}}\sup_{\mathbf{a}%
_{t}\in A^{t},\mathbf{b}_{t}\in B^{t}}\mathbb{E}\left[ \Delta V_{t}^{u}(%
\mathbf{b\,}|\,\mathbf{c})\right] \right) =0.
\end{equation*}

We have the following key decomposition:

\begin{proposition}
\label{p:K=DV+Reg}For every forecasting sequence $\mathbf{c}$ and sequence $%
\mathbf{b}$ we have%
\begin{equation*}
\mathcal{K}_{t}^{L^{u}}(\mathbf{c};\mathbf{b}\times \mathbf{c})=\text{%
\textsc{Reg}}_{t}^{u}(\mathbf{c};\mathbf{b}\times \mathbf{c})=\text{\textsc{%
Reg}}_{t}^{u}(\mathbf{c})+\Delta V_{t}^{u}(\mathbf{b\,}|\,\mathbf{c})
\end{equation*}%
for all $t\geq 1$, all state sequences $\mathbf{a}_{t}\in A^{t}$, all
sequences $\mathbf{b}_{t}\in B^{t}$, and all utility functions $u$ with
induced proper scoring rule $L^{u}.$
\end{proposition}

\begin{proof}
We have%
\begin{eqnarray*}
\mathcal{K}_{t}^{L^{u}}(\mathbf{c};\mathbf{b}\times \mathbf{c}) &=&\text{%
\textsc{Reg}}_{t}^{u}(\mathbf{c};\mathbf{b}\times \mathbf{c})=V_{t}^{u}(%
\mathbf{b}\times \mathbf{c})-U_{t}(\mathbf{c}) \\
&=&[V_{t}^{u}(\mathbf{b}\times \mathbf{c})-V_{t}^{u}(\mathbf{c})]+[V_{t}^{u}(%
\mathbf{c})-U_{t}(\mathbf{c})] \\
&=&\Delta V_{t}^{u}(\mathbf{b\,}|\,\mathbf{c})+\text{\textsc{Reg}}_{t}^{u}(%
\mathbf{c}).
\end{eqnarray*}
\end{proof}

\medskip

This yields

\begin{theorem}
\label{th: proper calibeat universal}A $\mathbf{b}$-based forecasting
procedure is proper complete $B$-calibeating if and only if it has universal
no regret and is universally $B$-subsuming.
\end{theorem}

For a $\mathbf{b}$-based procedure, the bound on regret must also be uniform
over the $\mathbf{b}$ sequences (which are used when generating the
forecasting sequence $\mathbf{c}$): the supremum in (\ref{eq:univ-regret})
is taken in addition over all $\mathbf{b}_{t}\in B^{t}$.

\medskip

\begin{proof}
Use Proposition \ref{p:K=DV+Reg}: since \textsc{Reg} and $\Delta V$ are
nonnegative, $\mathbb{E}\left[ \mathcal{K}_{t}^{L^{u}}(\mathbf{c};\mathbf{b}%
\times \mathbf{c})\right] \rightarrow 0$ if and only if $\mathbb{E}\left[
\Delta V_{t}^{u}(\mathbf{b\,}|\,\mathbf{c})\right] \rightarrow 0$ and $%
\mathbb{E}\left[ \text{\textsc{Reg}}_{t}^{u}(\mathbf{c})\right] \rightarrow
0 $ (with all convergences uniform).
\end{proof}

\medskip

Thus,%
\begin{equation*}
\begin{tabular}{|ccl|}
\hline
complete $\text{calibeating}$ & $\Longleftrightarrow $ & $\text{proper
complete calibeating\ }$ \\ 
& $\Longleftrightarrow $ & $\text{universal no regret\ }+\;\text{universal
subsuming}$ \\ \hline
\end{tabular}%
\end{equation*}

Comparing this with proper calibration, the additional requirement is
universal subsuming. Since the forecasting sequence $\mathbf{c}$ may well
incorporate knowledge of its own, the reference sequence $\mathbf{b}$ yields
no upper bound on how much $\mathbf{c}$ can improve upon it. What matters
for $\mathbf{c}$ is that it leave no decision value from $\mathbf{b}$
\textquotedblleft on the table\textquotedblright ---which is precisely
captured by the subsuming condition that there is no loss in ignoring $%
\mathbf{b}$ and using only $\mathbf{c}$. Moreover, we want this to hold for
every decision problem and utility function (which may well be unknown to
the forecaster)---universally (in addition to the uniformity over all $%
\mathbf{a}$ and $\mathbf{b}$ sequences). Proper complete calibeating is thus
a strong form of calibeating: it produces a forecast sequence $\mathbf{c}$
that, besides having no regret, exhausts all the decision value of the
reference sequence $\mathbf{b}$, universally over all decision problems with
bounded utility.

\medskip

Combining this with the results of Section \ref{sus:joint}, we obtain that
the following statements are all equivalent:

\begin{itemize}
\item complete calibeating

\item proper complete calibeating

\item calibration of the joint binning

\item proper calibration of the joint binning

\item universal no joint regret

\item universal no regret + universal subsuming
\end{itemize}

\medskip

\noindent In Appendix \ref{sus-a:improvement} we show that there is yet
another equivalent property when the reference sequence consists of
forecasts, or more generally induces decisions:

\begin{itemize}
\item universal exhaustive improvement
\end{itemize}

\noindent This says that the change in average utility from the reference
sequence to the calibeating sequence, which proper calibeating guarantees to
be nonnegative in the limit, becomes \textquotedblleft
exhaustive\textquotedblright\ for all bounded utilities if and only if the
calibeating is complete.

\medskip

\noindent \textbf{Remark. }Perfect\textbf{\ }complete calibeating says that $%
\bar{a}(b,c)=c$ for every nonempty $(b,c)$-bin, and thus $\bar{a}(b,c)=\bar{a%
}(c)=c$; this can be expressed as\footnote{%
Fix $t,$ and think of \textsc{a} as a random variable uniformly distributed
over the $t$ values $a_{1},...,a_{t}$; similarly for \textsc{b} and \textsc{c%
}.} $\mathbb{E}\left[ \text{\textsc{a}\thinspace }|\,\text{\textsc{b}},\text{%
\textsc{c}}\right] =\mathbb{E}\left[ \text{\textsc{a}}\,|\,\text{\textsc{c}}%
\right] =\,$\textsc{c}$.$ Viewing the elements of $A$ as the unit vectors of 
$C,$ the expectation of \textsc{a} is a vector whose coordinates are the
corresponding probabilities (frequencies), and so we have 
\begin{equation}
\mathbb{P}\left[ \,\text{\textsc{a}}=a\,|\,\text{\textsc{b}},\text{\textsc{c}%
\thinspace }\right] =\mathbb{P}\left[ \,\text{\textsc{a}}=a\,|\,\text{%
\textsc{c}\thinspace }\right] =\,\text{\textsc{c}}(a)
\label{eq:cc=subsume+calib}
\end{equation}%
for every $a\in A.$ The first equality says that \textsc{b} provides no
information on the distribution of \textsc{a} beyond that contained in 
\textsc{c} (i.e., \textsc{a} is independent of \textsc{b} given \textsc{c};
this is \textquotedblleft subsuming\textquotedblright ); the second, that 
\textsc{c} is calibrated.

Related notions are \textquotedblleft cross-calibration\textquotedblright\
(Feinberg and Stewart 2008), \textquotedblleft conditional-mean
encompassing\textquotedblright\ (Wooldridge 1990) and \textquotedblleft
forecast encompassing\textquotedblright\ (Chong and Hendry 1986).

\bigskip

In summary, we have demonstrated that \emph{complete calibeating} is the
\textquotedblleft ultimate\textquotedblright\ calibeating: it provides a
complete beating of the reference sequence $\mathbf{b}$, in the sense that
it extracts all information about states from $\mathbf{b}$ (which can
therefore be completely ignored), while being itself calibrated. Moreover,
in the non-perfect world with errors and asymptotic guarantees, these
properties are robust, i.e., uniform with respect to scoring rules and
utility functions.

\appendix

\section{Appendix}

\subsection{Scoring Rules\label{sus-a:scoring rules}}

We provide here further details on proper scoring rules. As in Section \ref%
{sus:scoring rules}, a scoring rule is given by $L_{A}:A\times C\rightarrow 
\mathbb{R}$, linearly extended to $L:C\times C\rightarrow \mathbb{R}$ by $%
L(d,c):=\sum_{a\in A}d(a)L_{A}(a,c)=d\cdot \mathbf{L}(c)$, where $\mathbf{L}%
(c)$ is the vector $(L_{A}(a,c))_{a\in A}$ in $\mathbb{R}^{A}$. The scoring
rule $L$ is proper if $L(d,c)\geq L(d,d)$ for every $c,d\in C$.

The following result is well known (see, e.g., Savage 1971; Gneiting \&
Raftery 2007).

\begin{proposition}
\label{p:proper}A scoring rule $L$ is proper if and only if there exists a
concave function $H:C\rightarrow \mathbb{R}$ and a supergradient selection $%
\mathbf{G}:C\rightarrow \mathbb{R}^{A}$ (i.e., for every $c\in C$ the vector 
$\mathbf{G}(c)$ is a supergradient of $H$ at\footnote{%
I.e., $H(d)\leq H(c)+(d-c)\cdot \mathbf{G}(c)$ for every $d\in C.$} $c$)
such that%
\begin{equation}
L(d,c)=H(c)+(d-c)\cdot \mathbf{G}(c)  \label{eq:S from H}
\end{equation}%
for every $c,d\in C.$
\end{proposition}

\begin{proof}
Assume that $L$ is proper. Then $H(d):=L(d,d)=\min_{c\in C}L(d,c)=\min_{c\in
C}d\cdot \mathbf{L}(c)$ is the minimum of linear functions of $d$, and thus $%
H$ is concave. The vector $\mathbf{L}(c)$ is a supergradient of $H$ at $c,$
because $H(d)-H(c)=L(d,d)-L(c,c)\leq L(d,c)-L(c,c)=d\cdot \mathbf{L}%
(c)-c\cdot \mathbf{L}(c)=(d-c)\cdot \mathbf{L}(c),$ where the inequality is
by properness.

Conversely, given a concave $H$ with a supergradient selection $\mathbf{G}$,
the supergradient inequality $H(d)\leq H(c)+(d-c)\cdot \mathbf{G}(c)$
becomes $L(d,d)\leq L(d,c)$ for the function $L$ that is defined by (\ref%
{eq:S from H}), and so $L$ is proper.
\end{proof}

\medskip

The concave function $H\equiv H^{L}$ given by $H(c):=L(c,c)$ for a proper
scoring rule $L$ is usually referred to as the $L$-\emph{entropy}.

\medskip

\noindent \textbf{Remark. }To avoid confusion: in the second part of the
proof the supergradient $\mathbf{G}(c)$ need \emph{not} be the vector $%
\mathbf{L}(c)$ that was used in the first part. The reason is that if $%
\mathbf{G}(c)$ is a supergradient then so is $\mathbf{G}(c)+\lambda \mathbf{1%
}$ for any real $\lambda $ (where $\mathbf{1}=(1,\ldots ,1)\in \mathbb{R}%
^{A} $), because on the domain $C$ we have $(d-c)\cdot \lambda \mathbf{1}=0$
for every $c,d\in C$ (and so formula (\ref{eq:S from H}) is not affected by
adding $\lambda \mathbf{1}$ to $\mathbf{G}(c)$). Since $c\cdot \mathbf{L}%
(c)=L(c,c)=H(c)$, we get the explicit relation\footnote{%
When $H$ is differentiable on a full-dimensional set in $\mathbb{R}^{A}$
that contains $C$ and $\mathbf{G}$ is the gradient $\nabla H$ of $H$, we
thus get $\mathbf{L}(c)=\nabla H(c)+(H(c)-c\cdot \nabla H(c))\mathbf{1}$.
This explains the formulas in the examples in the next section. In
particular, when $H$ is homogeneous of degree $1$ (as is the case for the $%
\alpha $-spherical scoring rules) we have $H(c)=c\cdot \nabla H(c)$ by
Euler's theorem, and so $\mathbf{L}=\nabla H$.} $\mathbf{L}(c)=\mathbf{G}%
(c)+(H(c)-c\cdot \mathbf{G}(c))\mathbf{1}$.

\medskip

The $L$\emph{-divergence} $D\equiv D^{L}:C\times C\rightarrow \mathbb{R}$ is
given by 
\begin{equation*}
D(d,c)\,:=\,L(d,c)-L(d,d),
\end{equation*}
and so $L$ is proper if and only if $D\geq 0.$ Geometrically, in this case
we have $D(d,c)=H(c)+(d-c)\cdot \mathbf{G}(c)-H(d)$ by (\ref{eq:S from H}),
and so the divergence $D(d,c)$ is equal to how much the tangent to the
concave function $H$ at $c$ (with slope $\mathbf{G}(c)$) is above $H$ at the
point $d$ (the \textquotedblleft Bregman divergence\textquotedblright ).

A scoring rule $L$ is \emph{bounded} if the functions $L_{A}(a,\cdot )$ for
all $a\in A$, and thus $L(d,\cdot )$ for all $d\in C$, are bounded; it is 
\emph{Lipschitz} if the functions $L_{A}(a,\cdot )$ for all $a\in A$, and
thus $L(d,\cdot )$ for all $d\in C$, are Lipschitz continuous (in the
forecast). It is convenient to state these conditions in terms of the vector
function $\mathbf{L}$ and the standard Euclidean norm $\left\Vert \cdot
\right\Vert $. Let $M$ be a finite constant; then:

\begin{itemize}
\item the scoring rule $L$ is $M$\emph{-bounded} if%
\begin{equation*}
\left\Vert \mathbf{L}(c)-\mathbf{L}(c^{\prime })\right\Vert \leq M
\end{equation*}%
for every $c,c^{\prime }\in C$; and

\item the scoring rule $L$ is $M$\emph{-Lipschitz} if%
\begin{equation*}
\left\Vert \mathbf{L}(c)-\mathbf{L}(c^{\prime })\right\Vert \leq M\left\Vert
c-c^{\prime }\right\Vert
\end{equation*}%
for all $c,c^{\prime }\in C$.\footnote{%
When the entropy $H$ is a so-called \textquotedblleft $M$-smooth%
\textquotedblright\ function, i.e., with an $M$-Lipschitz gradient, the
function $\mathbf{L}$ is $M\sqrt{|A|}$-Lipschitz (see the Remark above).}
\end{itemize}

We record some immediate useful implications.

\begin{lemma}
\label{l:O(diff^2)}Let $L$ be a proper scoring rule. Then

\begin{description}
\item[(i)] If $L$ is $M$-bounded then%
\begin{equation*}
0\leq D(d,c)\leq M\left\Vert c-d\right\Vert
\end{equation*}%
for every $c,d$ in $C$.

\item[(ii)] If $L$ is $M$-Lipschitz then%
\begin{equation*}
0\leq D(d,c)\leq M\left\Vert c-d\right\Vert ^{2}
\end{equation*}%
for every $c,d$ in $C$.

\item[(iii)] If $L$ is $M$-Lipschitz then%
\begin{equation*}
\left\vert D(d,c)-D(d,c^{\prime })\right\vert =\left\vert
L(d,c)-L(d,c^{\prime })\right\vert \leq M\left\Vert c-c^{\prime }\right\Vert
\end{equation*}%
for every $c,c^{\prime },d$ in $C$.
\end{description}
\end{lemma}

\begin{proof}
We have%
\begin{eqnarray*}
D(d,c) &\leq &D(d,c)+D(c,d) \\
&=&d\cdot \mathbf{L}(c)-d\cdot \mathbf{L}(d)+c\cdot \mathbf{L}(d)-c\cdot 
\mathbf{L}(c) \\
&=&(c-d)\cdot (\mathbf{L}(d)-\mathbf{L}(c))\leq \left\Vert c-d\right\Vert
\left\Vert \mathbf{L}(d)-\mathbf{L}(c)\right\Vert .
\end{eqnarray*}%
This yields (i) when $L$ is $M$-bounded and (ii) when $L$ is $M$-Lipschitz.
For (iii), using (\ref{eq:d*Lambda(c)}), the definition of $D$, and $%
\left\Vert d\right\Vert \leq 1$ for $d\in C$, gives:%
\begin{eqnarray*}
\left\vert D(d,c)-D(d,c^{\prime })\right\vert &=&\left\vert
L(d,c)-L(d,c^{\prime })\right\vert =\left\vert d\cdot \mathbf{L}(c)-d\cdot 
\mathbf{L}(c^{\prime })\right\vert \\
&\leq &\left\Vert d\right\Vert \left\Vert \mathbf{L}(c)-\mathbf{L}(c^{\prime
})\right\Vert \leq M\left\Vert c-c^{\prime }\right\Vert .
\end{eqnarray*}
\end{proof}

\subsubsection{Examples of Bounded Proper Scoring Rules\label%
{susus-a:scoring examples}}

We provide a number of classical examples of bounded proper scoring rules;
in fact, they are all strictly proper.

\medskip

\begin{itemize}
\item \emph{Quadratic}: 
\begin{eqnarray*}
L_{A}(a,c) &=&-2c(a)+\left\Vert c\right\Vert ^{2}, \\
L(d,c) &=&-2c\cdot d+\left\Vert c\right\Vert ^{2}, \\
H(c) &=&-\left\Vert c\right\Vert ^{2}, \\
D(d,c) &=&\left\Vert c-d\right\Vert ^{2}.
\end{eqnarray*}

\item $\alpha $-\emph{Spherical} for $\alpha >1$: 
\begin{eqnarray*}
L_{A}(a,c) &=&-\frac{c(a)^{\alpha -1}}{\left( \sum_{a}c(a)^{\alpha }\right)
^{\frac{\alpha -1}{\alpha }}}, \\
L(d,c) &=&-\frac{\sum_{a}d(a)c(a)^{\alpha -1}}{\left( \sum_{a}c(a)^{\alpha
}\right) ^{\frac{\alpha -1}{\alpha }}}, \\
H(c) &=&-\left( \sum_{a}c(a)^{\alpha }\right) ^{\frac{1}{\alpha }%
}=-\left\Vert c\right\Vert _{\alpha }, \\
D(d,c) &=&-\frac{\sum_{a}d(a)c(a)^{\alpha -1}}{\left( \sum_{a}c(a)^{\alpha
}\right) ^{\frac{\alpha -1}{\alpha }}}+\left( \sum_{a}d(a)^{\alpha }\right)
^{\frac{1}{\alpha }}
\end{eqnarray*}%
(for $\alpha \geq 2$ the scoring rule is Lipschitz).

\item $\alpha $-\emph{Power} (\emph{Tsallis}) for $\alpha >1$:%
\begin{eqnarray*}
L_{A}(a,c) &=&-\frac{1}{\alpha -1}c(a)^{\alpha -1}+\frac{1}{\alpha }%
\sum_{a}c(a)^{\alpha }, \\
L(d,c) &=&-\frac{1}{\alpha -1}\sum_{a}d(a)c(a)^{\alpha -1}+\frac{1}{\alpha }%
\sum_{a}c(a)^{\alpha }, \\
H(c) &=&-\frac{1}{\alpha (\alpha -1)}\sum_{a}c(a)^{\alpha }=-\frac{1}{\alpha
(\alpha -1)}\left( \left\Vert c\right\Vert _{\alpha }\right) ^{\alpha }, \\
D(d,c) &=&\frac{1}{\alpha }\sum_{a}c(a)^{\alpha }+\frac{1}{\alpha (\alpha -1)%
}\sum_{a}d(a)^{\alpha }-\frac{1}{\alpha -1}\sum_{a}d(a)c(a)^{\alpha -1}
\end{eqnarray*}%
(for $\alpha \geq 2$ the scoring rule is Lipschitz; for $\alpha =2$ it is
the quadratic score divided by $2$).
\end{itemize}

\subsection{Approximate Decomposition for Local Binnings\label%
{sus-a:decompose-delta}}

We give here the proof of Theorem \ref{th:decompose-delta} in Section \ref%
{susus:fractional}.

\begin{proof}[Proof of Theorem \protect\ref{th:decompose-delta}]
Replacing every $c_{s}$ and $\bar{c}_{t}(i)$ in bin $i$ with $y^{i}$ yields
scores\footnote{%
The refinement score is not affected since it does not depend on the
forecasts. For simplicity we drop the subscript $t$ from $n_{t},\bar{a}_{t},$
and $\bar{c}_{t}$.}%
\begin{eqnarray*}
\widehat{\mathcal{B}}_{t}^{L}(\mathbf{c}) &%
%TCIMACRO{\TeXButton{:=}{{\;:=\;}}}%
%BeginExpansion
{\;:=\;}%
%EndExpansion
&\sum_{i\in I}\left( \frac{n(i)}{t}\right) \sum_{s=1}^{t}\left( \frac{%
f_{s}(i)}{n(i)}\right) D(a_{s},y^{i})\text{\ \ \ and} \\
\widehat{\mathcal{K}}_{t}^{L}(\mathbf{c};\mathbf{f}) &%
%TCIMACRO{\TeXButton{:=}{{\;:=\;}}}%
%BeginExpansion
{\;:=\;}%
%EndExpansion
&\sum_{i\in I}\left( \frac{n(i)}{t}\right) D(\bar{a}(i),y^{i}).
\end{eqnarray*}%
The proof proceeds in two steps: first, we show that $\widehat{\mathcal{B}}$
and $\widehat{\mathcal{K}}$ are close to $\mathcal{B}$ and $\mathcal{K}$,
respectively; second, we establish the exact decomposition\footnote{%
Since positive fractions of a forecast $c_{s}$ may be allocated to several
bins, and thus replaced by \emph{different} $y^{i}$, there is no single
replacement of the forecasting sequence $\mathbf{c}$ that would give Step 2
by applying directly the decomposition Theorem \ref{th:S-decompose}.} $%
\widehat{\mathcal{B}}-\widehat{\mathcal{K}}=\mathcal{R}$.

\emph{Step 1.} 
\begin{eqnarray*}
\left\vert \mathcal{B}_{t}^{L}(\mathbf{c})-\widehat{\mathcal{B}}_{t}^{L}(%
\mathbf{c})\right\vert &\leq &M\delta \text{\ \ \ and} \\
\left\vert \mathcal{K}_{t}^{L}(\mathbf{c};\mathbf{f})-\widehat{\mathcal{K}}%
_{t}^{L}(\mathbf{c};\mathbf{f})\right\vert &\leq &M\delta .
\end{eqnarray*}%
Indeed, all $c_{s}$ in bin $i$ (i.e., with $f_{s}(i)>0$) satisfy $\left\Vert
c_{s}-y^{i}\right\Vert \leq \delta $, and thus their average $\bar{c}%
(i)\equiv \bar{c}_{t}(i)$ satisfies $\left\Vert \bar{c}(i)-y^{i}\right\Vert
\leq \delta $ as well. Therefore $\left\vert
D(a_{s},c_{s})-D(a_{s},y^{i})\right\vert \leq M\delta $ and $\left\vert D(%
\bar{a}(i),\bar{c}(i))-D(\bar{a}(i),y^{i})\right\vert \leq M\delta $ by (\ref%
{eq:D-D-L}). Averaging the former over $s$ and $i$ yields the inequality for 
$\mathcal{B}$, and averaging the latter over $i$ yields the inequality for $%
\mathcal{K}$.

\emph{Step 2.}%
\begin{equation*}
\widehat{\mathcal{B}}_{t}^{L}(\mathbf{c})=\widehat{\mathcal{K}}_{t}^{L}(%
\mathbf{c};\mathbf{f})+\mathcal{R}_{t}^{L}(\mathbf{f}).
\end{equation*}%
Indeed, we have%
\begin{eqnarray*}
\widehat{\mathcal{B}}_{t}^{L}(\mathbf{c}) &=&\sum_{i\in I}\left( \frac{n(i)}{%
t}\right) \sum_{s=1}^{t}\left( \frac{f_{s}(i)}{n(i)}\right) L(a_{s},y^{i})-%
\mathcal{H}_{t}^{L} \\
&=&\sum_{i\in I}\left( \frac{n(i)}{t}\right) L(\bar{a}(i),y^{i})-\mathcal{H}%
_{t}^{L},
\end{eqnarray*}%
and%
\begin{eqnarray*}
\widehat{\mathcal{K}}_{t}^{L}(\mathbf{c};\mathbf{f}) &=&\sum_{i\in I}\left( 
\frac{n(i)}{t}\right) \left[ L(\bar{a}(i),y^{i})-L(\bar{a}(i),\bar{a}(i))%
\right] \\
&=&\sum_{i\in I}\left( \frac{n(i)}{t}\right) L(\bar{a}(i),y^{i})-\sum_{i\in
I}\left( \frac{n(i)}{t}\right) H(\bar{a}(i)).
\end{eqnarray*}%
Subtracting proves the claim by (\ref{eq:R-H-H}).

Combining Step 1 and Step 2 yields the result.
\end{proof}

\subsection{Refinement of Binnings\label{sus-a:refined-refinement}}

We give here the missing proofs of the results of Section \ref{susus:refined
refinement}.

\begin{proof}[Proof of Proposition \protect\ref{p:refine-S}]
Let $I$ and $J$ be the sets of bins of $\mathbf{f}$ and $\mathbf{g}$,
respectively. It suffices to prove the claim when $J$ has only one bin; we
then apply it to each $j$-bin separately and average over $j$ to get the
general result. Dropping the subscript $t$ for convenience from $\bar{a}%
_{t}(i)$ and $n_{t}(i)$, and letting $\bar{a}$ denote the overall average of
the states (i.e., the average in the single bin in $J$), we have:%
\begin{eqnarray*}
\frac{1}{t}\sum_{i\in I}\sum_{s=1}^{t}f_{s}(i)L(a_{s},\bar{a}(i))
&=&\sum_{i\in I}\left( \frac{n(i)}{t}\right) \sum_{s=1}^{t}\left( \frac{%
f_{s}(i)}{n(i)}\right) L(a_{s},\bar{a}(i)) \\
&=&\sum_{i\in I}\left( \frac{n(i)}{t}\right) L(\bar{a}(i),\bar{a}%
(i))=\sum_{i\in I}\left( \frac{n(i)}{t}\right) H(\bar{a}(i)) \\
&\leq &H\left( \sum_{i}\frac{n(i)}{t}\bar{a}(i)\right) =H(\bar{a}) \\
&=&L(\bar{a},\bar{a})=\frac{1}{t}\sum_{s=1}^{t}L(a_{s},\bar{a}),
\end{eqnarray*}%
where the inequality is by the concavity of the function $H\equiv H^{L}$.
Subtracting $\mathcal{H}_{t}^{L}$ from both sides yields the desired
inequality.
\end{proof}

\medskip

\begin{proof}[Proof of Corollary \protect\ref{c:refine-S}]
The Decomposition Theorem \ref{th:S-decompose}, applied to each one of $%
\mathbf{c}$, $\mathbf{i}$, and $\mathbf{j}$ (all of which refine $\mathbf{c}$%
), yields $\mathcal{B}^{L}(\mathbf{c})=\mathcal{R}^{L}(\mathbf{c})+\mathcal{K%
}_{t}^{L}(\mathbf{c})=\mathcal{R}^{L}(\mathbf{i})+\mathcal{K}^{L}(\mathbf{c};%
\mathbf{i})=\mathcal{R}^{L}(\mathbf{j})+\mathcal{K}^{L}(\mathbf{c};\mathbf{j}%
)$; apply Proposition \ref{p:refine-S}.
\end{proof}

\subsection{On the Proof of Proper Calibeating by a Proper Calibrated
Procedure \label{sus-a:bxc}}

Our proof of Theorem \ref{th:beat-by-calib S}, which establishes proper
calibeating by a proper calibrated procedure, is based on complete
calibeating, i.e., the stronger construct of calibeating the joint binning
(see Section \ref{sus:joint}). As noted in the remark at the end of the
section (cf. the example in Section \ref{sus:example}), calibeating together
with calibration does not guarantee proper calibeating. In this appendix we
provide evidence for the need for complete calibeating. We do so by keeping
the state averages in all joint bins fixed, and varying the relative
frequencies of these bins. In this setting, we show that complete
calibeating is necessary to guarantee that proper calibeating always follows
from calibeating by a calibrated forecast. Our main tool is Theorem \ref%
{th:U} in Section \ref{susus-a:bxc general} below, a general result that may
be of independent interest.

Thus, fix the finite set of bins $B$ of the sequence $\mathbf{b}$, the
finite set of bins $D$ of our forecasting sequence $\mathbf{c}$, and the
state average $\bar{a}(b,d)$ of each joint $(b,d)$-bin, and allow the
(relative) frequencies $\lambda (b,d)$ of the bins to vary.\footnote{%
We are thus looking at a \textquotedblleft snapshot\textquotedblright\ of
history that considers only empirical distributions and ignores the specific
sequences and the time horizon.} For instance, repeating twice each period
of a certain $(b,d)$-bin amounts to doubling this $\lambda (b,d)$ while
keeping all the rest unchanged.\footnote{%
Up to renormalization.} For clarity, we assume that all calibration and
calibeating errors are exactly zero. Given bin frequencies $\lambda $, we
take our sequence $\mathbf{c}\equiv \mathbf{c}_{\lambda }$ to be (perfectly)
calibrated, and so $c=\bar{a}(\cdot ,d)=\sum_{b}\lambda (b,d)\bar{a}%
(b,d)/\sum_{b}\lambda (b,d)$ for each bin\footnote{%
While in general the $\mathbf{c}$ binning may be coarser than the $\mathbf{d}
$ binning (because some $\bar{a}(\cdot ,d)$ averages may turn out to be
equal), they coincide for generic $\lambda $ where these averages are all
distinct. Our proof below handles all cases.} $d$.

\begin{theorem}
\label{th:bxc}Assume that the matrix $(\bar{a}(b,d))_{b\in B,d\in D}$ of
state averages has more than two distinct entries. Then the following
statements are equivalent:

\begin{description}
\item[(i)] For every bin-frequency matrix $\lambda $, if the calibrated
forecasting sequence $\mathbf{c}\equiv \mathbf{c}_{\lambda }$ calibeats $%
\mathbf{b}$ then it properly calibeats $\mathbf{b}$.

\item[(ii)] For every bin-frequency matrix $\lambda $, if the calibrated
forecasting sequence $\mathbf{c}\equiv \mathbf{c}_{\lambda }$ calibeats $%
\mathbf{b}$ then it completely calibeats $\mathbf{b}$.
\end{description}
\end{theorem}

One way to interpret this result is as follows: if there are bin frequencies 
$\lambda $ where the calibrated sequence $\mathbf{c}_{\lambda }$ is
calibeating $\mathbf{b}$ but it does \emph{not} completely calibeat $\mathbf{%
b}$ (and so (ii) does not hold), then there are other bin frequencies $%
\lambda ^{\prime }$ such that the calibrated $\mathbf{c}_{\lambda ^{\prime
}} $ calibeats $\mathbf{b}$ but does \emph{not} properly calibeat $\mathbf{b}
$ (i.e., (i) does not hold). In short, for calibeating to entail proper
calibeating no matter what the bin frequencies are, one needs calibeating to
entail complete calibeating.

The proof will be provided in Section \ref{susus-a:proof}, after stating and
establishing the general result (Theorem \ref{th:U}).

\subsubsection{A General Result\label{susus-a:bxc general}}

We establish a general result that proves Theorem \ref{th:bxc}.

Let $I$ and $J$ be finite sets, $X=(x_{ij})_{i\in I,j\in J}$ a matrix whose
entries are $m$-dimensional real vectors, i.e., $x_{ij}\in \mathbb{R}^{m}$,
and let $W=(w_{ij})_{i\in I,j\in J}$ be a weight matrix, i.e., $w_{ij}\geq 0$
for all $i,j,$ and $\sum_{i\in I}\sum_{j\in J}w_{ij}=1$. Define%
\begin{eqnarray*}
w_{i\cdot } &%
%TCIMACRO{\TeXButton{:=}{{\;:=\;}}}%
%BeginExpansion
{\;:=\;}%
%EndExpansion
&\sum_{j\in J}w_{ij} \\
w_{\cdot j} &%
%TCIMACRO{\TeXButton{:=}{{\;:=\;}}}%
%BeginExpansion
{\;:=\;}%
%EndExpansion
&\sum_{i\in I}w_{ij} \\
r_{i} &%
%TCIMACRO{\TeXButton{:=}{{\;:=\;}}}%
%BeginExpansion
{\;:=\;}%
%EndExpansion
&\bar{x}_{i\cdot }=\sum_{j\in J}\frac{w_{ij}}{w_{i\cdot }}x_{ij} \\
c_{j} &%
%TCIMACRO{\TeXButton{:=}{{\;:=\;}}}%
%BeginExpansion
{\;:=\;}%
%EndExpansion
&\bar{x}_{\cdot j}=\sum_{i\in I}\frac{w_{ij}}{w_{\cdot j}}x_{ij}
\end{eqnarray*}%
(thus, $w_{i\cdot }$ and $w_{\cdot j}$ are the marginals, and $r_{i}$ and $%
c_{j}$ the row and column averages; the values of $r_{i}$ and $c_{j}$ for
rows and columns with zero weight will not matter).\footnote{%
For the application to calibeating, the set of rows $I$ is $B$ (the range of 
$\mathbf{b}$, the \textquotedblleft reference forecasts\textquotedblright ),
the set of columns $J$ is $D$ (the range of $\mathbf{c}$, the forecasts
used), the weights $w$ are the bin frequencies $\lambda $, and the matrix
entries $x$ are the state averages $\bar{a}$.}

For every concave function $F:\mathbb{R}^{m}\rightarrow \mathbb{R}$ (it
suffices for $F$ to be defined on the compact convex set $\mathrm{conv}%
\{x_{ij}:i\in I,j\in J\}$) let%
\begin{eqnarray*}
E_{W}(F) &%
%TCIMACRO{\TeXButton{:=}{{\;:=\;}}}%
%BeginExpansion
{\;:=\;}%
%EndExpansion
&\sum_{i\in I}\sum_{j\in J}w_{ij}F(x_{ij}) \\
R_{W}(F) &%
%TCIMACRO{\TeXButton{:=}{{\;:=\;}}}%
%BeginExpansion
{\;:=\;}%
%EndExpansion
&\sum_{i\in I}w_{i\cdot }F(r_{i}) \\
C_{W}(F) &%
%TCIMACRO{\TeXButton{:=}{{\;:=\;}}}%
%BeginExpansion
{\;:=\;}%
%EndExpansion
&\sum_{j\in J}w_{\cdot j}F(c_{j})
\end{eqnarray*}%
(these are the overall average, average by rows, and average by columns,
respectively). We are interested in the inequalities $C\leq R$; more
precisely, when does the inequality $C(Q)\leq R(Q)$ for a quadratic concave
function $Q$ imply $C(F)\leq R(F)$ for all concave functions $F$.

The matrix $X$ is $W$\emph{-column-constant}\footnote{%
Interpreting $W$ as a probability measure, this means that $X$ is $W$-almost
surely column-constant.} if the restriction of $X$ to the support of $W$ has
constant columns, i.e., $x_{ij}=x_{i^{\prime }j}$ whenever $w_{ij}>0$ and $%
w_{i^{\prime }j}>0$; the matrix $X$ is \emph{column-constant }if it has
constant columns, i.e., $x_{ij}=x_{i^{\prime }j}$ for every $i,i^{\prime
}\in I$ and $j\in J$ (thus, $X$ is column-constant if and only if it is $W$%
-column-constant for every $W$). Similarly for $W$\emph{-row-constant} ($%
x_{ij}=x_{ij^{\prime }}$ whenever $w_{ij}>0$ and $w_{ij^{\prime }}>0$) and 
\emph{row-constant} ($x_{ij}=x_{ij^{\prime }}$ for every $i\in I$ and $%
j,j^{\prime }\in J$). For instance, if $W$ is a diagonal matrix then every $%
X $ is both $W$-column-constant and $W$-row-constant.

\begin{proposition}
\label{p:0}Let $W$ be a weight matrix. Then:

\begin{description}
\item[(a)] $C_{W}(F)\geq E_{W}(F)$ and $R_{W}(F)\geq E_{W}(F)$ for every
concave $F$.

\item[(b1)] If $X$ is $W$-column-constant then $C_{W}(F)=E_{W}(F)$ for every
concave $F$.

\item[(b2)] If $X$ is $W$-row-constant then $R_{W}(F)=E_{W}(F)$ for every
concave $F$.

\item[(c1)] $C_{W}(G)=E_{W}(G)$ for some \emph{strictly} concave function $G$
if and only if $X$ is $W$-column-constant (and then $C_{W}(F)=E_{W}(F)$ for
every concave $F$ by (b1)).

\item[(c2)] $R_{W}(G)=E_{W}(G)$ for some \emph{strictly} concave function $G$
if and only if $X$ is $W$-row-constant (and then $R_{W}(F)=E_{W}(F)$ for
every concave $F$ by (b2)).
\end{description}
\end{proposition}

\begin{proof}
(a) The concavity of $F$ yields%
\begin{equation}
F(c_{j})\geq \sum_{i}\frac{w_{ij}}{w_{\cdot j}}F(x_{ij})  \label{eq:concave}
\end{equation}%
for every $j$; multiplying by $w_{\cdot j}$ and summing over $j$ then gives%
\begin{equation}
C_{W}(F)=\sum_{j}w_{\cdot j}F(c_{j})\geq
\sum_{j}\sum_{i}w_{ij}F(x_{ij})=E_{W}(F).  \label{eq:R>=E}
\end{equation}

(b1) If $X$ is $W$-column-constant then $x_{ij}=c_{j}$ for every $i,j$ with $%
w_{ij}>0$, and so we get equality in (\ref{eq:concave}), and thus in (\ref%
{eq:R>=E}). Similarly for (b2).

(c1) Equality in (\ref{eq:R>=E}), and thus in (\ref{eq:concave}) for every $%
j $ with $w_{\cdot j}>0$, holds for a strictly concave $G$ if and only if
all the entries $x_{ij}$ in column $j$ that have a positive weight $w_{ij}>0$
must be equal; that is, $X$ is $W$-column-constant. Similarly for (c2).
\end{proof}

\medskip

A useful strictly concave function $F$ is the quadratic $Q(z)=-\left\Vert
z\right\Vert ^{2}$.

We will say that a matrix $X$ is \emph{non-degenerate} if it has more than
two distinct entries.

\begin{theorem}
\label{th:U}Let $X$ be a non-degenerate matrix. Then the following three
statements are equivalent.

\begin{description}
\item[(U1)] For every weight matrix $W$, if $C_{W}(Q)\leq R_{W}(Q)$ then $%
C_{W}(F)\leq R_{W}(F)$ for every concave function $F$.

\item[(U2)] For every weight matrix $W$, if $C_{W}(Q)\leq R_{W}(Q)$ then $%
C_{W}(Q)=E_{W}(Q)$.

\item[(U3)] For every weight matrix $W$, if $C_{W}(Q)\leq R_{W}(Q)$ then $%
C_{W}(F)=E_{W}(F)\leq R_{W}(F)$ for every concave function $F$.
\end{description}
\end{theorem}

Now (U3) trivially implies (U2) (since $Q$ is concave), and (U2) implies
(U1) (by Proposition \ref{p:0} (c1): since $Q$ is strictly concave, $%
C_{W}(Q)=E_{W}(Q)$ implies that $X$ is $W$-column-constant and then $%
C_{W}(F)=E_{W}(F)\leq R_{W}(F)$ for every concave $F$). It thus remains to
show that \textbf{(U1) implies (U3)}.

We first consider an easy special case of matrices $X$ that are row-constant
(whether degenerate or not); see Proposition \ref{p:row}. We then address
the substantial case of non-degenerate matrices $X$ that are not
row-constant. In this case, the statements hold if and only if the matrix $X$
is column-constant; see Theorem \ref{th:column}.

\begin{proposition}
\label{p:row}Let $X$ be a row-constant matrix. Then the statements (U1),
(U2), and (U3) of Theorem \ref{th:U} are equivalent.
\end{proposition}

\begin{proof}
As noted above, we need to show that \textbf{(U1) implies (U3).}

Assume (U1). If $X$ is row-constant then $C_{W}(Q)\geq E_{W}(Q)=R_{W}(Q)$
for all $W$ by Proposition \ref{p:0} (a) and (b2), and so if $W$ is such
that $C_{W}(Q)\leq R_{W}(Q)$ then we have equality $%
C_{W}(Q)=E_{W}(Q)=R_{W}(Q)$; since $Q$ is strictly concave, by Proposition %
\ref{p:0} (c1), $X$ is (also) $W$-column-constant, and then $%
C_{W}(F)=E_{W}(F)=R_{W}(F)$ for every concave $F$ by Proposition \ref{p:0}
(b1) and (b2).~Thus (U1) implies (U3).
\end{proof}

\begin{theorem}
\label{th:column}Let $X$ be a non-degenerate matrix that is not
row-constant. Then the statements (U1), (U2), and (U3) of Theorem \ref{th:U}
and (UC) below are equivalent.

\begin{description}
\item[(UC)] The matrix $X$ is column-constant.
\end{description}

\noindent Moreover, in this case $C_{W}(F)=E_{W}(F)\leq R_{W}(F)$ (for every
concave $F$), and thus $C_{W}(Q)\leq R_{W}(Q)$, hold for every $W$.
\end{theorem}

\begin{proof}
We will show that (U1) implies (UC) (this is the substantial part of the
proof), and (UC) implies (U3) and the \textquotedblleft
moreover\textquotedblright\ statement.

\textbf{(U1) implies (UC). }Assume (U1).

$\bullet $ \emph{Claim 1.} For every $W$ such that $C_{W}(Q)\leq R_{W}(Q)$,
the row averages $r_{i}$ must lie in the convex hull of the column averages $%
c_{j}$; i.e., letting $I_{0}:=\{i\in I:w_{i\cdot }>0\}$ and $J_{0}:=\{j\in
J:w_{\cdot j}>0\}$, we have $r_{i}\in \mathrm{conv}\{c_{j}:j\in J_{0}\}$ for
every $i\in I_{0}$.

\emph{Proof. }Let $\Gamma :=\mathrm{conv}\{c_{j}:j\in J_{0}\}$ and assume
that $r_{i}\notin \Gamma $ for some $i$ with $w_{i\cdot }>0.$ Since $\Gamma $
is a closed convex set, by the Separating Hyperplane Theorem there are $v\in 
\mathbb{R}^{m}$ and $\alpha \in \mathbb{R}$ such that $v\cdot r_{i}>\alpha
\geq v\cdot c_{j}$ for every $j\in J_{0}$. For every $n$ let $F_{n}$ be the
following function: $F_{n}(z):=Q(z)-n(\max \{v\cdot z-\alpha ,0\})^{2}$; the
function $F_{n}$ is concave (as the sum of two concave functions). Now $%
C_{W}(F_{n})=C_{W}(Q)$ for all $n$ (because $F_{n}(c_{j})=Q(c_{j})$ for all $%
j$ with $w_{\cdot j}>0$), whereas $R_{W}(F_{n})\rightarrow -\infty $ as $%
n\rightarrow \infty $ (because $v\cdot r_{i}-\alpha >0$ and so $%
F_{n}(r_{i})\rightarrow -\infty $); therefore, for $n$ large enough the
inequality $C_{W}(F_{n})\leq R_{W}(F_{n})$ fails---contradicting (U1).

$\bullet $ \emph{Claim 2.} Every $2\times 2$ submatrix of $X$ has at most $2$
distinct entries.

\emph{Proof}. Let 
\begin{equation*}
X^{0}=\left[ 
\begin{array}{cc}
a & b \\ 
d & \cdot%
\end{array}%
\right]
\end{equation*}
be a $2\times 2$ submatrix of $X$ with $a,b,d$ distinct (the fourth entry
will get weight $0$ and so will not matter). We distinguish two cases,
according to whether or not $d$ lies in the open interval $(a,b)$ (i.e., on
the straight line through $a$ and $b$, between $a$ and $b$).

Case A: $d\notin (a,b)$. Let 
\begin{equation*}
W=\left[ 
\begin{array}{cc}
(1-\varepsilon )/2 & (1-\varepsilon )/2 \\ 
\varepsilon & 0%
\end{array}%
\right]
\end{equation*}%
for $\varepsilon >0$. As $\varepsilon \rightarrow 0$ we get $%
R_{W}(Q)\rightarrow -\left\Vert (a+b)/2\right\Vert ^{2}$ and $%
C_{W}(Q)\rightarrow -\left\Vert a\right\Vert ^{2}/2-\left\Vert b\right\Vert
^{2}/2,$ and so $a\neq b$ yields $C_{W}(Q)<R_{W}(Q)$ for small enough $%
\varepsilon >0$; but $c_{1}\rightarrow a,$ $c_{2}=b$, and $r_{2}=d$, and so $%
r_{2}\notin \mathrm{conv}\{c_{1},c_{2}\}$ for small enough $\varepsilon >0$,
a contradiction to Claim 1.

Case B: $d\in (a,b)$. Restricting to the $1$-dimensional space (the line)
that contains $a,b,$ and $d$ and using the coordinate system $\lambda
a+(1-\lambda )b\mapsto \lambda $ yields the matrix\footnote{%
The change of coordinates does not affect the sign of $R_{W}(Q)-C_{W}(Q)$;
it amounts to using the quadratic $Q_{1}(z)=-\left\Vert z-b\right\Vert
^{2}/\left\Vert a-b\right\Vert ^{2}$ instead of $Q(z)=-\left\Vert
z\right\Vert ^{2}$. Every inequality between $C_{W}$ and $R_{W}$ that holds
for $Q$ holds also for $Q_{1}$, in fact for every $F=\lambda Q+G$ with $%
\lambda >0$ and $G$ affine. Indeed, $C_{W}(G)=E_{W}(G)=R_{W}(G)$ since $G$
is affine, and so $C_{W}(\lambda Q+G)-R_{W}(\lambda Q+G)=\lambda
(C_{W}(Q)-R_{W}(Q))$ since $C_{W}(F)$ and $R_{W}(F)$ are linear in $F.$} 
\begin{equation*}
\left[ 
\begin{array}{cc}
1 & 0 \\ 
\delta & \cdot%
\end{array}%
\right]
\end{equation*}%
with $0<\delta <1.$ Let 
\begin{equation*}
W=\frac{1}{3(1+\delta )}\left[ 
\begin{array}{cc}
1+2\delta & 1-\delta \\ 
1+2\delta & 0%
\end{array}%
\right] ,
\end{equation*}%
then a straightforward computation yields $R_{W}(Q)-C_{W}(Q)=\delta
(1-\delta )(1+2\delta )/(6(2+\delta ))>0$, but $c_{1}=(1+\delta )/2,$ $%
c_{2}=0$, and $r_{1}=(1+2\delta )/(2+\delta )>c_{1}>c_{2}$, and so $%
r_{1}\notin \mathrm{conv}\{c_{1},c_{2}\}$, a contradiction to Claim 1.

This completes the proof of Claim 2.

$\bullet $ \emph{Claim 3.} $X$ is column-constant (i.e., (UC)).

\emph{Proof}. Assume that $X$ is not column-constant; let $a\neq b$ be two
distinct entries in, say, column 1. Then Claim 2 implies that every row that
has $a$ or $b$ in column 1 must contain only $a$ and $b$. If column 1 were
to contain only $a$ and $b$, then the entire matrix would contain only $a$
and $b$, a contradiction to the non-degeneracy of $X$. Therefore column 1
must contain another distinct entry $d$. But then each row with $a$ in
column 1 cannot contain $b$ (again by Claim 2, using the two rows with $a$
and $d$ in column 1), and so the row is a constant $a$ row; similarly, each
row with $b$ in column 1 cannot contain $a$, and so it is a constant $b$
row. Carrying out the same argument, but now with the pair of distinct
entries $a$ and $d$, shows that every row with $d$ in column 1 must be a
constant $d$ row; this holds for every $d$, and so all rows are constant
rows---a contradiction to our assumption that $X$ is not row-constant. This
completes the proof of Claim 3; thus, (U1) implies (UC).

\textbf{(UC) implies (U3) and the \textquotedblleft
moreover\textquotedblright\ statement.}

If $X$ is column-constant then $C_{W}(F)=E_{W}(F)\leq R_{W}(F)$ for every $W$
and every concave $F$ by Proposition \ref{p:0} (a) and (b1).
\end{proof}

\medskip

\noindent \textbf{Remarks. }\emph{(a) }The result continues to hold even if
we restrict the concave functions $F$ to those that generate bounded proper
scoring rules (see Proposition \ref{p:proper}), in particular those with
bounded gradients in the relevant domain (namely, the compact convex set
that contains all $x_{ij}$); indeed, our construction uses only such
functions $F_{n}$ (see the proof of Claim 1 above).

\emph{(b)} The result is false when $X$ has only two distinct entries. For
example, let 
\begin{equation*}
X=\left[ 
\begin{array}{cc}
1 & 0 \\ 
0 & 1%
\end{array}%
\right] .
\end{equation*}%
For every weight matrix $W$ (with positive row and column weight sums) we
have $r_{1}\geq c_{1}$ iff $w_{11}/(w_{11}+w_{12})\geq
w_{11}/(w_{11}+w_{21}) $ iff $w_{12}\leq w_{21}$ iff $w_{22}/(w_{22}+w_{12})%
\geq w_{22}/(w_{22}+w_{21})$ iff $c_{2}\geq r_{2}$, and so $r_{1},r_{2}$ are
either both inside the interval $[c_{1},c_{2}]$, or both strictly outside.
Because the corresponding weighted averages of the $r_{i}$ and the $c_{j}$
are equal, in the first case we have $C_{W}(F)\leq R_{W}(F)$ for every
concave $F$, and in the second case we have $C_{W}(F)>R_{W}(F)$ for all
strictly concave $F.$ Therefore, (U1) holds. However, (U2) does not hold:
take for instance all weights to be $1/4$, then $C_{W}(Q)=R_{W}(Q)=-1/4,$
whereas $E_{W}(Q)=-1/2$.

\emph{(c)} There are degenerate non-row-constant matrices for which Theorem %
\ref{th:U} holds; for instance, an $n\times n$ identity matrix $X$ for $%
n\geq 3$.

\subsubsection{Proof\label{susus-a:proof}}

We prove the main result of this Appendix.

\medskip

\begin{proof}[Proof of Theorem \protect\ref{th:bxc}]
We use Theorem \ref{th:U} above, with $X=(\bar{a}(b,d))_{b\in B,d\in D}$ the
matrix of state averages and $W=\lambda \equiv (\lambda (b,d))_{b\in B,d\in
D}$ the matrix of bin frequencies.

The construction of the perfectly calibrated $\mathbf{c}\equiv \mathbf{c}%
_{\lambda }$ yields $\mathcal{K}^{L}(\mathbf{c};\mathbf{c})=\mathcal{K}^{L}(%
\mathbf{c};\mathbf{d})=0,$ and thus $\mathcal{B}^{L}(\mathbf{c})=\mathcal{R}%
^{L}(\mathbf{c})=\mathcal{R}^{L}(\mathbf{d})$ (by the Decomposition Theorem %
\ref{th:S-decompose}), for every $L\in \mathcal{L}.$ Let $\lambda (b,\cdot
):=\sum_{d}\lambda (b,d)$ and $\lambda (\cdot ,d):=\sum_{b}\lambda (b,d)$
denote the marginal frequencies, and $\bar{a}(b,\cdot )$ and $\bar{a}(\cdot
,d)$ the marginal state averages; for every $L\in \mathcal{L}$ we have%
\begin{eqnarray*}
\mathcal{R}^{L}(\mathbf{b}) &=&\sum_{b}\lambda (b,\cdot )H^{L}(\bar{a}%
(b,\cdot ))-\mathcal{H}^{L}=R_{W}(H^{L})-\mathcal{H}^{L} \\
\mathcal{B}^{L}(\mathbf{c}) &=&\mathcal{R}^{L}(\mathbf{c})=\mathcal{R}^{L}(%
\mathbf{d})=\sum_{d}\lambda (\cdot ,d)H^{L}(\bar{a}(\cdot ,d))-\mathcal{H}%
^{L}=C_{W}(H^{L})-\mathcal{H}^{L},\text{\ \ \ and} \\
\mathcal{R}^{L}(\mathbf{b}\times \mathbf{c}) &\geq &\mathcal{R}^{L}(\mathbf{b%
}\times \mathbf{d})=\sum_{b}\sum_{d}\lambda (b,d)H^{L}(\bar{a}(b,d))-%
\mathcal{H}^{L}=E_{W}(H^{L})-\mathcal{H}^{L},
\end{eqnarray*}%
(by (\ref{eq:R-H-H})), where $H^{L}$ is the concave $L$-entropy function,
and $\mathcal{H}^{L}$ is the average $L$-entropy of the states (which does
not matter for the comparisons below).

Therefore:

\begin{itemize}
\item $\mathbf{c}$ is calibeating $\mathbf{b}$ if and only if $C_{W}(Q)\leq
R_{W}(Q)$ (because for the quadratic scoring rule we have $%
H^{L}(z)=Q(z)=-\left\Vert z\right\Vert ^{2}$);

\item $\mathbf{c}$ is proper calibeating $\mathbf{b}$ if and only if $%
C_{W}(H^{L})\leq R_{W}(H^{L})$ for every bounded proper scoring rule%
\footnote{%
The normalization of $L$ no longer matters since the errors are assumed to
be exactly $0$.} $L\in \mathcal{L}$, which is equivalent to $C_{W}(F)\leq
R_{W}(F)$ for every concave $F$ that generates a bounded proper scoring rule
(see Proposition \ref{p:proper} and Remark (a) in Section \ref{susus-a:bxc
general}: Theorem \ref{th:U} remains valid with (U1) restricted to this
class); and

\item $\mathbf{c}$ is calibeating $\mathbf{b}\times \mathbf{d}$ if and only
if $C_{W}(Q)=E_{W}(Q)$ (see (J1) in Remark (a) following Theorem \ref%
{th:beat-by-calib S}).
\end{itemize}

Thus (U1) (more precisely, its above restricted version) of Theorem \ref%
{th:U} is precisely (i), and (U2) implies (ii) (because $\mathcal{R}^{L}(%
\mathbf{b}\times \mathbf{d})\leq \mathcal{R}^{L}(\mathbf{b}\times \mathbf{c}%
) $, and so if $\mathbf{c}$ is calibeating $\mathbf{b}\times \mathbf{d}$
then it calibeats $\mathbf{b}\times \mathbf{c}$).\footnote{%
When the $\mathbf{c}$ and $\mathbf{d}$ binnings coincide, (U2) is precisely
(ii).} Since (U1) and (U2) are equivalent by Theorem \ref{th:U}, and (ii)
implies (i) (as established in Section \ref{sus:joint}), it follows that (i)
and (ii) are equivalent.
\end{proof}

\subsection{Proper Calibeating and Universal Utility Improvement\label%
{sus-a:improvement}}

In the decision-making under uncertainty setup of Section \ref{s:utility},
we consider here the effect of calibeating on utility. The reference
sequence is now taken to consist of actual forecasts, i.e., $B$ is a finite
subset of $C=\Delta (A)$, and we study the increase in average utility that
can be guaranteed by calibeating (the Remark at the end of this subsection
shows that the results generalize beyond the case of a forecasting reference
sequence).

\begin{proposition}
\label{p:U-U}Let $\mathbf{b}$ and $\mathbf{c}$ be forecasting sequences, $%
\mathbf{i}$ a binning sequence that refines $\mathbf{b}$, and $u$ a utility
function with induced proper scoring rule $L^{u}$. Then%
\begin{equation*}
U_{t}(\mathbf{c})-U_{t}(\mathbf{b})=\text{\textsc{Reg}}_{t}^{u}(\mathbf{b};%
\mathbf{i})+\left( \mathcal{R}_{t}^{L^{u}}(\mathbf{i})-\mathcal{B}%
_{t}^{L^{u}}(\mathbf{c})\right) .
\end{equation*}
\end{proposition}

\begin{proof}
By Proposition \ref{p:reg=K} and the definition of regret, we have%
\begin{equation*}
U_{t}(\mathbf{c})=-\mathcal{B}_{t}^{L^{u}}(\mathbf{c})-\mathcal{H}%
_{t}^{L^{u}}
\end{equation*}%
and%
\begin{equation*}
U_{t}(\mathbf{b})=V_{t}^{u}(\mathbf{i})-\text{\textsc{Reg}}_{t}^{u}(\mathbf{b%
};\mathbf{i})=-\mathcal{R}_{t}^{L^{u}}(\mathbf{i})-\mathcal{H}_{t}^{L^{u}}-%
\text{\textsc{Reg}}_{t}^{u}(\mathbf{b};\mathbf{i}).
\end{equation*}%
Subtracting yields the result.
\end{proof}

\medskip

Thus, if $\mathbf{c}$ is $L^{u}$-calibeating $\mathbf{b}$ (i.e., $\mathcal{R}%
_{t}^{L^{u}}(\mathbf{b})-\mathcal{B}_{t}^{L^{u}}(\mathbf{c})\geq 0$), best
replying to the forecasts $c_{s}$ instead of the forecasts $b_{s}$ (i.e.,
taking the decisions $x^{\ast }(c_{s})$ instead of $x^{\ast }(b_{s})$)
yields a nonnegative improvement in average utility, at least as large as
the regret of $\mathbf{b}$.

\begin{corollary}
\label{c:u-calibeat}Let $B\subseteq C$, and let $\zeta $ be a $\mathbf{b}$%
-based forecasting procedure that is proper $(\varepsilon ,B)$-calibeating.
Then\footnote{%
The term $U_{t}(\mathbf{b})$ is constant and can be taken out of the
expectation.}%
\begin{equation}
\mathbb{E}\left[ U_{t}(\mathbf{c})-U_{t}(\mathbf{b})\right] \geq \text{%
\textsc{Reg}}_{t}^{u}(\mathbf{b})-\varepsilon ^{2}-o(1)
\label{eq:u-calibeat}
\end{equation}%
as $t\rightarrow \infty $, uniformly over all state sequences $\mathbf{a}$,
all forecasting sequences $\mathbf{b}$, and all utility functions $u$ with
induced proper scoring rule $L^{u}$ that is $1$-bounded.
\end{corollary}

\begin{proof}
Apply Proposition \ref{p:U-U} with $\mathbf{i}=\mathbf{b}$ and use the
proper calibeating inequality $\mathbb{E}\left[ \mathcal{B}_{t}^{L^{u}}(%
\mathbf{c})-\mathcal{R}_{t}^{L^{u}}(\mathbf{b})\right] \leq \varepsilon
^{2}+o(1).$
\end{proof}

\medskip

Is the regret of $\mathbf{b}$ the largest attainable improvement? An
immediate observation is that a general upper bound is provided by the
regret of $\mathbf{b}$ with respect to the joint binning $\mathbf{b}\times 
\mathbf{c}$:

\begin{proposition}
\label{p:U<=U+Reg}Let $B\subseteq C$, and let $\zeta $ be a $\mathbf{b}$%
-based forecasting procedure. Then%
\begin{equation*}
U_{t}(\mathbf{c})-U_{t}(\mathbf{b})\leq \text{\textsc{Reg}}_{t}^{u}(\mathbf{b%
};\mathbf{b}\times \mathbf{c})
\end{equation*}%
for all $t\geq 1$, all state sequences $\mathbf{a}_{t}\in A^{t}$, all
sequences $\mathbf{b}_{t}\in B^{t}$, and all utility functions $u$.
\end{proposition}

\begin{proof}
By Proposition \ref{p:U-U} and the Decomposition Theorem \ref{th:S-decompose}
we have%
\begin{eqnarray}
U_{t}(\mathbf{c})-U_{t}(\mathbf{b}) &=&\text{\textsc{Reg}}_{t}^{u}(\mathbf{b}%
;\mathbf{b}\times \mathbf{c})+\mathcal{R}_{t}^{L^{u}}(\mathbf{b}\times 
\mathbf{c})-\mathcal{B}_{t}^{L^{u}}(\mathbf{c})  \notag \\
&=&\text{\textsc{Reg}}_{t}^{u}(\mathbf{b};\mathbf{b}\times \mathbf{c})-%
\mathcal{K}_{t}^{L^{u}}(\mathbf{c};\mathbf{b}\times \mathbf{c})\leq \text{%
\textsc{Reg}}_{t}^{u}(\mathbf{b};\mathbf{b}\times \mathbf{c}).
\label{eq:improv-joint}
\end{eqnarray}
\end{proof}

\medskip

This upper bound is in fact achieved by (proper) complete calibeating, i.e.,
(proper) calibeating the joint binning; see Theorems \ref{th:proper complete
calibeating} and \ref{th:beat-by-calib S}.

\begin{corollary}
\label{c:u-calibeat-joint}Let $B\subseteq C$, and let $\zeta $ be a $\mathbf{%
b}$-based forecasting procedure that is proper complete $(\varepsilon ,B)$%
-calibeating. Then%
\begin{equation*}
\mathbb{E}\left[ \text{\textsc{Reg}}_{t}^{u}(\mathbf{b};\mathbf{b}\times 
\mathbf{c})\right] -\varepsilon ^{2}-o(1)\leq \mathbb{E}\left[ U_{t}(\mathbf{%
c})-U_{t}(\mathbf{b})\right] \leq \mathbb{E}\left[ \text{\textsc{Reg}}%
_{t}^{u}(\mathbf{b};\mathbf{b}\times \mathbf{c})\right]
\end{equation*}%
and in addition we have%
\begin{equation*}
0\leq \mathbb{E}\left[ \text{\textsc{Reg}}_{t}^{u}(\mathbf{c})\right] \leq 
\mathbb{E}\left[ \text{\textsc{Reg}}_{t}^{u}(\mathbf{c};\mathbf{b}\times 
\mathbf{c})\right] \leq \varepsilon ^{2}+o(1);
\end{equation*}%
these inequalities hold as $t\rightarrow \infty $ uniformly over all state
sequences $\mathbf{a}$, all forecasting sequences $\mathbf{b}$, and all
utility functions $u$ with induced proper scoring rule $L^{u}$ that is $1$%
-bounded.
\end{corollary}

\begin{proof}
Apply Proposition \ref{p:U-U} with $\mathbf{i}=\mathbf{b}\times \mathbf{c}$
and use the proper-calibeating inequality $\mathbb{E}\left[ \mathcal{B}%
_{t}^{L^{u}}(\mathbf{c})-\mathcal{R}_{t}^{L^{u}}(\mathbf{b}\times \mathbf{c})%
\right] \leq \varepsilon ^{2}+o(1)$; for the second statement, use $\mathcal{%
K}_{t}^{L^{u}}(\mathbf{c})\leq \mathcal{K}_{t}^{L^{u}}(\mathbf{c};\mathbf{b}%
\times \mathbf{c})=\mathcal{B}_{t}^{L^{u}}(\mathbf{c})-\mathcal{R}%
_{t}^{L^{u}}(\mathbf{b}\times \mathbf{c})$.
\end{proof}

\medskip

This suggests defining the notion of an \textquotedblleft exhaustive%
\footnote{%
We point out that \textquotedblleft exhaustive\textquotedblright\ is a local
rather than global notion.} improvement,\textquotedblright\ namely, an
average utility improvement that achieves the upper bound of Proposition \ref%
{p:U<=U+Reg}. We require this uniformly over histories and utilities: a $%
\mathbf{b}$-based forecasting procedure $\zeta $ is \emph{universally }$%
(\varepsilon ,B)$\emph{-exhaustively-improving} if 
\begin{equation*}
\varlimsup_{t\rightarrow \infty }\left( \sup_{u\in \mathcal{U}_{1}}\sup_{%
\mathbf{a}_{t}\in A^{t},\mathbf{b}_{t}\in B^{t}}\mathbb{E}\left[ \text{%
\textsc{Reg}}_{t}^{u}(\mathbf{b};\mathbf{b}\times \mathbf{c})-\left( U_{t}(%
\mathbf{c})-U_{t}(\mathbf{b})\right) \right] \right) \leq \varepsilon ^{2}.
\end{equation*}%
Since \textsc{Reg}$_{t}^{u}(\mathbf{b};\mathbf{b}\times \mathbf{c})-\left(
U_{t}(\mathbf{c})-U_{t}(\mathbf{b})\right) =\mathcal{K}_{t}^{L^{u}}(\mathbf{c%
};\mathbf{b}\times \mathbf{c})$ by (\ref{eq:improv-joint}), we have:

\begin{theorem}
\label{th:proper calibeat universal improving}Let $B\subseteq C$. A $\mathbf{%
b}$-based forecasting procedure $\zeta $ is proper complete $(\varepsilon
,B) $-calibeating if and only if it is universally $(\varepsilon ,B)$%
-exhaustively-improving.
\end{theorem}

This allows us to add \emph{universal exhaustive improvement} to the list of
equivalent properties of (proper) complete calibeating\footnote{%
With the proviso that the reference sequence consists of forecasts (or, more
generally, induces decisions; see the Remark below).} at the end of Section %
\ref{susus:u-subsuming}.

\medskip

\noindent \textbf{Remark. }The analysis above applies with no change to any
setup where the reference sequence induces a decision sequence by a fixed
mapping $x^{\#}:B\rightarrow X$ (fixed over time, though it may depend on $u$%
). In this case the average utility of $\mathbf{b}$ is well defined by $%
U_{t}(\mathbf{b})=(1/t)\sum_{s\leq t}u(a_{s},x^{\#}(b_{s})).$ Since the
exact definition of $U_{t}(\mathbf{b})$ plays no role in the results in this
subsection (it only needs to be defined, and to be bounded from above by the
decision value $V_{t}^{u}(\mathbf{b})$, so that \textsc{Reg}$_{t}^{u}(%
\mathbf{b};\mathbf{i})=V_{t}^{u}(\mathbf{i})-U_{t}(\mathbf{b})\geq 0$),
everything carries through. Two leading examples: first, as considered
above, when each $b$ is a forecast (i.e., $B\subset C$) and $x^{\#}$ is a $u$%
-best reply mapping $x^{\ast }$. Second, when each $b$ is already a
decision, i.e., $B$ is a finite subset of $X$ (and $x^{\#}$ is the identity
map).

\end{document}